%% file: Fat-Robots-CA.tex
\documentclass[11pt]{article}
\usepackage{amsmath, amssymb}
\usepackage{color}
\usepackage{graphicx}
\usepackage{epstopdf}
\usepackage[margin = 1in]{geometry}
\usepackage{times}

\newtheorem{theorem}{Theorem}
\newtheorem{lemma}[theorem]{Lemma}
\newtheorem{definition}{Definition}
\newtheorem{corollary}[theorem]{Corollary}
\newenvironment{proof}{\noindent{\bf Proof:}}

\newcommand{\remove}[1]{}
\newcommand{\CH}{onCH}
\newcommand{\con}{CH}

\begin{document}

\title{A Distributed Algorithm for Gathering Many Fat Mobile Robots\\ in the Plane}
\date{}
\author{
Chrysovalandis Agathangelou~~~~~~ Chryssis Georgiou~~~~~~ Marios Mavronicolas\\
{\normalsize Department of Computer Science}\\
{\normalsize University of Cyprus}\\
{\normalsize CY-1678 Nicosia, Cyprus}
}
\maketitle

\begin{abstract}
In this work we consider the problem of gathering autonomous robots
in the plane. In particular, we consider non-transparent unit-disc robots
(i.e., fat) in an asynchronous setting. Vision is the only mean of coordination.
Using a state-machine representation we formulate the gathering problem and
develop a distributed algorithm that solves the problem for any number of robots.

The main idea behind our algorithm is for the robots to reach a configuration in which all
the following hold: (a) The robots' centers form a convex hull in which all robots are on the convex,
(b) Each robot can see all other robots, and (c) The configuration is connected, that is, every robot touches another robot and all robots together form a connected formation.
We show that starting from any initial configuration, the robots, making only local decisions and coordinate
by vision, eventually reach such a configuration and terminate, yielding a solution to the gathering problem. 
\end{abstract}

%%%%%%%%%%%%%%%%%%%%%%%%%%%%%%%% INTRO %%%%%%%%%%%%%%%%%%%%%%%%%%%%%%%%%%%%%%%%%%%%%

\input{intro}

%%%%%%%%%%%%%%%%%%%%%%%%%%%%%%%%% MODEL %%%%%%%%%%%%%%%%%%%%%%%%%%%%%%%%%%%%%%%%%%%%%%

\section{Model and Definitions}

Our model of computation is a formalization of the one presented in~\cite{Pelc09} (with the additional assumption of chilarity);
our formalism follows the one from~\cite{Attiya04}. 
%{\bf Chrysovalandis: Additionally we assume that robots have chirality (i.e. when robots agree on the orientation of the axes). In our model we assume that robots can distinguish left from right. However, this is different
%than having a common coordinates system or having a compass.}

\paragraph{Robots:}
We assume $n$ asynchronous fault-free robots that can move in straight lines on 
the (infinite) plane. The robots are {\em fat}~\cite{Pelc09}: they are closed 
unit discs. They are identical, anonymous and indistinguishable. They do not 
have access to any global coordination system, but we assume that the robots agree on the
orientation of the axes (i.e., per~\cite{robotsbook} they have chilarity\footnote{In other words, we assume that they have
a common understanding of what is left or right. Note this is a weaker assumption than having
a common coordinate system or having a compass~\cite{robotsbook}.}). Robots are equipped with a 
360-degree-angle  vision devise (e.g., camera) that enables the robots to 
take snapshots of the plane. The vision devise can capture any point of the plane 
(has unlimited range) provided there is no obstacle (e.g., another robot). We assume that robots know $n$.

\paragraph{Geometric configuration:} A {\em geometric configuration} is a vector
${\cal G} = (c_1, c_2,\ldots,c_n)$ where each $c_i$ represents the center of the position
of robot $r_i$ on the plane. Informally speaking, a configuration can be viewed 
as a snapshot of the robots on the plane. Note that the fact that robots are fat 
prohibits the formation of a configuration in which any two robots share more than 
a point (on the perimeter of their unit discs) in the plane.  

We say that a geometric configuration ${\cal G}$ is {\em connected}, if between any 
two points of any two robots there exists a polygonal line each of whose points belongs to some robot. Informally, a configuration is connected
if every robot touches another robot (i.e., their circles representing the robots are tangent)
and all robots together form a connected formation. 

\paragraph{Visibility and fully visible configuration:} We say that point $p$ in the 
plain is {\em visible} by a robot $r_i$ (or equivalently, $r_i$ can see $p$) if there 
exists a point $p_i$ in the circle bounding robot $r_i$ such that the straight 
segment $(p_i,p)$ does not contain any point 
of any other robot. From this it follows that a robot $r_i$ can see another robot $r_j$ 
if there exists at least one point on the bounding circle of $r_j$ that is visible by $r_i$. 

Given a geometric configuration ${\cal G}$, if a robot $r_i$ can see all other robots,
then we say that robot $r_i$ has {\em full visibility} in ${\cal G}$. If all robots have full visibility in ${\cal G}$, then we say that configuration ${\cal G}$ is {\em fully visible}. 
%(Note that if a robot $r_i$ can see another robot $r_j$ then $r_i$ can compute the 
%center of $r_j$ (since $r_i$ can see a non-zero arc of the bounding cycle of $r_j$).)
%If $V_i = {\cal G}$ (robot $r_i$ can see all other robots) then we say that robot 

\paragraph{Robots' states:}
Formally, each robot $r_i$ is modeled as a (possibly infinite) state machine 
with state set $S_i$; $i$ is the index of robot $r_i$ (used only for reference purposes).
Each set $S_i$ contains five states: {\bf Wait}, {\bf Look}, {\bf Compute}, {\bf Move}, 
and {\bf Terminate}. Initially each robot is in state {\bf Wait}. State {\bf Terminate}
is a terminating state: once a robot reaches this state it does not take any further
steps. We now describe each state:
\begin{itemize}
\item 
In state {\bf Wait}, robot $r_i$ is idling.
In addition, the robot has no memory of the steps occurred prior entering this state
(that is, every time a robot gets into state {\bf Wait}, it looses any recollection
of past steps -- robots are history oblivious). 
\item 
In state {\bf Look}, robot $r_i$ takes a snapshot of the plane and identifies the robots 
that are visible to it. We denote by $V_i$ the set of the centers of the robots that are visible to robot $r_i$ when it takes a snapshot in configuration ${\cal G}$. That is, $V_i \subseteq {\cal G}$ is the {\em local view} of robot $r_i$ in configuration ${\cal G}$. Note that this view does not change in subsequent configurations unless the robot takes a new snapshot. In a nutshell, in this state, the robot takes as an input a configuration ${\cal G}$ and outputs the local view $V_i\subseteq {\cal G}$.  
\item 
In state {\bf Compute}, robot $r_i$ runs a local algorithm, call it $A_i$, that takes
as an input the local view $V_i$ (that is, the output of the previous state {\bf Look})
and outputs a point $p$ in the plane. This point is specified from $V_i$, hence 
we will write $p = A_i(V_i)$. If $A_i$ outputs the special point $\bot$, then the robot's
state changes into state {\bf Terminate}. Otherwise it changes into state {\bf Move}
(intuitively, in this case $p$ is the point that the center of the robot will move to).
Note that it is possible for $p=c_i$, that is, the robot might decide not to move.
\item
In state {\bf Move}, robot $r_i$ starting from its current position, called {\em start} point, moves on a straight line towards point $A_i(V_i)$
(as calculated in state {\bf Compute}). We call $A_i(V_i)$ the {\em target} point of
$r_i$. If during the 
motion the robot touches some other robot (i.e., the circles representing these robots become
tangent) it stops and the robot's state changes into state {\bf Wait}. As we discuss next, the adversary may also stop a robot at any point before reaching its target point. Again, in this case, the robot's state changes into state {\bf Wait}.
If the robot finds no obstacles or it is not stopped by the adversary then it eventually reaches its target point (its center is placed on $A_i(V_i)$) and its state changes into state {\bf Wait}.
\end{itemize}

\paragraph{State configuration:} A {\em state configuration} is a 
vector ${\cal S} = (s_1,s_2,\ldots,s_n)$ where each $s_i$ represents the state of robot
$r_i$. An {\em initial} state configuration is a configuration ${\cal S}$ in which
each $s_i$ is an initial state of robot $r_i$ (that is, $\forall i\in [1,n],~s_i = {\bf Wait}$).
Similarly, a {\em terminal} state configuration is a configuration ${\cal S}$ in which
each $s_i$ is a terminating state of robot $r_i$ (that is, $\forall i\in [1,n],~s_i = {\bf Terminate}$).

\paragraph{Robot configuration:} 
A {\em robot configuration} is a vector ${\cal R} = (\langle s_1,c_1\rangle,\ldots,\langle s_n,c_n\rangle)$ where each pair $\langle s_i,c_i\rangle$ represents the state of robot $r_i$
and the position of its center on the plane. (Informally, a robot configuration 
is the combination of a geometric configuration with the corresponding state configuration.)  

\paragraph{Adversary and events:}
We model asynchrony as events caused by an online and omniscient adversary.
The adversary can control the speeds of the robots, it can stop moving robots, 
and it may cause moving robots to collide, provided that their trajectories 
have an intersection point. 
Specifically, we consider the following events (state transitions):
\begin{itemize}
\item [] {\em Look}$(r_i)$: This event causes robot $r_i$ that is in state 
{\bf Wait} to get into state {\bf Look}.
\item [] {\em Compute}$(r_i)$: This event causes robot $r_i$ that is in state
{\bf Look} to get into state {\bf Compute}.
\item [] {\em Done}$(r_i)$: This event causes robot $r_i$ that is in state
{\bf Compute} and its local algorithm $A_i$ has returned the special point $\bot$, to 
get into the terminating state {\bf Terminate}.
\item [] {\em Move}$(r_i)$: This event causes robot $r_i$ that is in state 
{\bf Compute} and its local algorithm $A_i$ has returned a point other than $\bot$, 
to get into state {\bf Move}.
\item [] {\em Stop}$(r_i)$: This event causes robot $r_i$ that is in state
{\bf Move} to get into state {\bf Wait}. Robot $r_i$ is stopped at 
some point in the straight segment between its start point and its target 
point $A_i(V_i)$ (under a constraint
discussed next).
\item [] {\em Collide}$(R)$: This event causes a subset of the robots $R$ that are in 
state {\bf Move} and their trajectories have an intersecting point to collide
(i.e., their circles representing the robots become tangent). Note that $2\leq|R|\leq n$
(two or more robots could collide between them but only one collusion occurs per a {\em Collide} event).
Also, other robots that are in state {\bf Move} could be stopped (without colluding
with other robots). All affected robots are now in state {\bf Wait}.
\item [] {\em Arrive}$(r_i)$: This event causes robot $r_i$ that is in 
state {\bf Move} to arrive at its target point and change its
state into {\bf Wait}.      
\end{itemize}
Note that events {\em Look}$(r_i)$, {\em Move}$(r_i)$, {\em Stop}$(r_i)$ and 
{\em Arrive}$(r_i)$ may also cause robots (other than $r_i$) that are in state {\bf Move} 
to remain in that state, but on a different position on the plane (along
their trajectories, and closer to their destination).

Figure~\ref{fig:cycle} depicts a cycle of the state transitions of a 
robot $r_i$; it is understood that for event collide($R$), $r_i\in R$.

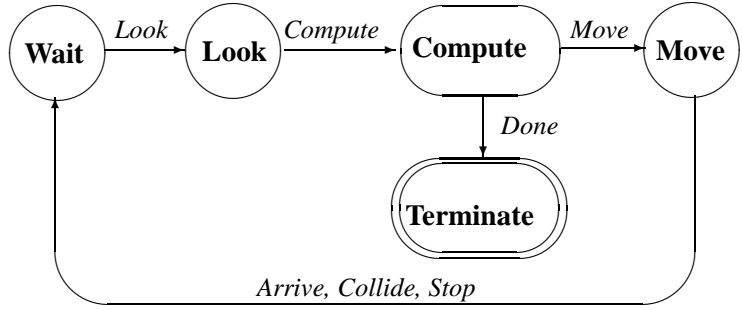
\begin{figure}[t]
\setlength{\unitlength}{3947sp}%
\begingroup\makeatletter\ifx\SetFigFont\undefined%
\gdef\SetFigFont#1#2#3#4#5{%
  \reset@font\fontsize{#1}{#2pt}%
  \fontfamily{#3}\fontseries{#4}\fontshape{#5}%
  \selectfont}%
\fi\endgroup%
\centering
\begin{picture}(3924,1909)(215,-1292)
%  METADATA <id>32</id> 
\thinlines
{\color[rgb]{0,0,0}\put(501, 14){\vector( 0, 1){0}}
\put(2174, 14){\oval(3346,2594)[bl]}
\put(2174, 26){\oval(4694,2618)[br]}
}%
%  METADATA <id>4</id> 
{\color[rgb]{0,0,0}\thinlines
\put(515,295){\circle{584}}
}%
%  METADATA <id>8</id> 
{\color[rgb]{0,0,0}\put(1620,310){\circle{584}}
}%
%  METADATA <id>10</id> 
{\color[rgb]{0,0,0}\put(3172,311){\oval(1000,562)}
}%
%%{\color[rgb]{0,0,0}\put(2688,307)%%{\circle{584}}
%%}%
%  METADATA <id>16</id> 
{\color[rgb]{0,0,0}\put(4500,318){\circle{584}}
}%
%  METADATA <id>22</id> 
{\color[rgb]{0,0,0}\thinlines
\put(3165,-676){\oval(1000,562)}
}%
%%{\color[rgb]{0,0,0}\put(2689,-720)%%{\circle{584}}
%%}%
%  METADATA <id>24</id> 
{\color[rgb]{0,0,0}\put(3165,-680){\oval(1102,626)}
}%
%%{\color[rgb]{0,0,0}\put(2685,-727)%%{\circle{706}}
%%}%
%  METADATA <id>26</id> 
{\color[rgb]{0,0,0}\put(814,313){\vector( 1, 0){500}}
}%
%  METADATA <id>29</id> 
{\color[rgb]{0,0,0}\put(1939,313){\vector( 1, 0){700}}
}%
%  METADATA <id>30</id> 
{\color[rgb]{0,0,0}\put(3680,326){\vector( 1, 0){525}}
}%
%%{\color[rgb]{0,0,0}\put(3027,326){\vector( %%1, 0){500}}
%%}%
%  METADATA <id>31</id> 
{\color[rgb]{0,0,0}\put(3189, 26){\vector( 0,-1){382}}
}%
%  METADATA <id>5</id> 
\put(310,232){\makebox(0,0)[lb]{\smash{{\SetFigFont{11}{14.4}{\rmdefault}{\mddefault}{\updefault}{\color[rgb]{0,0,0}{\bf Wait}}%
}}}}
%  METADATA <id>9</id> 
\put(1425,244){\makebox(0,0)[lb]{\smash{{\SetFigFont{11}{14.4}{\rmdefault}{\mddefault}{\updefault}{\color[rgb]{0,0,0}{\bf Look}}%
}}}}
%  METADATA <id>11</id> 
\put(2740,269){\makebox(0,0)[lb]{\smash{{\SetFigFont{11}{14.4}{\rmdefault}{\mddefault}{\updefault}{\color[rgb]{0,0,0}{\bf Compute}}%
}}}}
%  METADATA <id>17</id> 
\put(4274,255){\makebox(0,0)[lb]{\smash{{\SetFigFont{11}{14.4}{\rmdefault}{\mddefault}{\updefault}{\color[rgb]{0,0,0}{\bf Move}}%
}}}}
%  METADATA <id>23</id> 
\put(2700,-783){\makebox(0,0)[lb]{\smash{{\SetFigFont{11}{14.4}{\rmdefault}{\mddefault}{\updefault}{\color[rgb]{0,0,0}{\bf Terminate}}%
}}}}
%  METADATA <id>34</id> 
\put(1936,389){\makebox(0,0)[lb]{\smash{{\SetFigFont{10}{14.4}{\rmdefault}{\mddefault}{\updefault}{\color[rgb]{0,0,0}{\em Compute}}%
}}}}
\put(876,389){\makebox(0,0)[lb]{\smash{{\SetFigFont{10}{14.4}{\rmdefault}{\mddefault}{\updefault}{\color[rgb]{0,0,0}{\em Look}}%
}}}}
\put(3736,389){\makebox(0,0)[lb]{\smash{{\SetFigFont{10}{14.4}{\rmdefault}{\mddefault}{\updefault}{\color[rgb]{0,0,0}{\em Move}}%
}}}}
\put(1770,-1230){\makebox(0,0)[lb]{\smash{{\SetFigFont{10}{14.4}{\rmdefault}{\mddefault}{\updefault}{\color[rgb]{0,0,0}{\em Arrive,~Collide,~Stop}}%
}}}}
%  METADATA <id>35</id> 
\put(3300,-211){\makebox(0,0)[lb]{\smash{{\SetFigFont{10}{14.4}{\rmdefault}{\mddefault}{\updefault}{\color[rgb]{0,0,0}{\em Done}}%
}}}}
\end{picture}%
\caption{A cycle of the state transitions of robot $r_i$.}
\label{fig:cycle}
\end{figure}

\paragraph{Execution:} An {\em execution fragment} is an alternating sequence of
robot configurations and events. Formally, an execution fragment $\alpha$
is a  (finite or infinite) sequence of 
${\cal R}_0,e_1,{\cal R}_1,e_2,\ldots$, where each ${\cal R}_k$ is a robot 
configuration and each $e_k$ is an event. If $\alpha$ 
is finite, then it ends in a configuration. An {\em execution} is 
an execution fragment where ${\cal R}_0$ is an initial configuration.

\paragraph{Liveness conditions:} 
%Since the system is asynchronous, to guarantee termination (and hence eventual solution to %the problem under consideration) 
We impose the following liveness conditions (they are basically  restrictions on the adversary):

\begin{enumerate}
\item In an infinite execution, each robot %is allowed to perform 
may take infinitely 
many steps. 
%that is, we consider {\em admissible} executions. (This implies that
%robots do not fail.)
\item %The adversary must allow each robot, 
During a {\bf Move} event, each robot traverses 
at least a distance $\delta>0$ unless its target point is closer than $\delta$. 
Formally, each robot $r_i$ traverses at least a distance 
$\min\{dist_i(start,target),\delta\}$, where $dist_i(start,target)$ denotes the
distance between the start and target points of robot $r_i$. Parameter $\delta$ is 
not known to the robots (or to their local algorithms). 
\end{enumerate}
 
%We will be referring to the adversary under these constraints as {\em $\delta$-adversary}.
%Note that 

\paragraph{Gathering:} We now state the problem we consider in this work: 

\begin{definition}[Gathering problem]
In any execution, there is a connected, fully visible, terminal robot configuration.
\end{definition}

%%%%%%%%%%%%%%%%%%%%%%%%%%%%%%%%%% Geometric Functions %%%%%%%%%%%%%%%%%%%%%%%%%%%%%%  

\section{Geometric Functions}
In this section we present a collection of functions that
perform geometric calculations. These functions are used
by the robots' local algorithm as shown in the next section.
In this section we present these functions in a general manner,
with reference to the centers of unit discs on the plane (that is,
not necessarily for robots).
After the presentation of each function, we give some
insight on how this function is used by the robots' local algorithm.    

\subsection{Function {\tt On-Convex-Hull}}
\label{subsec:FOnCH}

\sloppy{We denote by %$\langle c_1,c_2,\ldots,c_m \rangle_{ch}$ 
$CH(c_1,c_2,\ldots,c_m)$ the {\em convex hull}
formed by points $c_1,c_2,\ldots,c_m$, and by $\CH(c_1,c_2,\ldots,c_m)\subseteq \{c_1,c2,\ldots,c_m\}$
the set of points that are {\em on} the convex hull.} 
Then, function {\tt On-Convex-Hull} solves the following algorithmic
problem:\\

\hrule\vspace{3pt}
\noindent {\sc On Convex Hull}\vspace{3pt}
\hrule\vspace{5pt}
\noindent {\bf Input:} A set of $m$ points $c_1,c_2,\ldots,c_m$ and an additional 
point $c$.\vspace{2pt}\\
{\bf Output:} {\sc yes} if $c\in \CH(c_1,c_2,\ldots,c_m)$ %$\langle c_1,c_2,\ldots,c_m\rangle_{ch}$, 
otherwise {\sc no}.\vspace{3pt}
\hrule\vspace{1em}

\sloppy{Function {\tt On-Convex-Hull} involves the computation of the convex hull
formed by points $c_1,c_2,\ldots,c_m$ and a check whether $c$ is one of the 
points on the convex hull.} The function returns, besides {\sc yes} or {\sc no},
also set $\CH(c_1,c_2,\ldots,c_m)$. 
%\subseteq \{c_1,c2,\ldots,c_m\}$ that are {\em on} $CH(c_1,c_2,\ldots,c_m)$. 
%$\langle c_1,c_2,\ldots,c_m\rangle_{ch}$. 
This function can easily be implemented using, for example, Graham's Convex Hull Algorithm~\cite{Graham72}. 

\paragraph{Insight:} This function is called by a robot $r$ with center $c$. The
$m\leq n$ points are the centers of the robots that robot $r$ can see (its local view)
in the current configuration. The robot wishes to check whether its center belongs 
on the convex hull formed by the centers of the robots in its local view. 
In the case of full visibility ($m=n$) robot $r$ can check whether itself as well as 
all other robots are on the convex hull. If this is the case, then by the definition 
of full visibility and of the convex hull, all robots can potentially see all other robots.

\subsection{Function {\tt Move-to-Point}}

Function {\tt Move-to-Point} solves the following algorithmic problem:\\

\hrule\vspace{3pt}
\noindent {\sc Move to Point}\vspace{3pt}
\hrule\vspace{5pt}
\noindent {\bf Input:} Two points $c_1$ and $c_2$ and a positive integer $m$.\vspace{2pt}\\
{\bf Output:} Point $\mu$ defined as follows: Consider the straight
segment $\overline{c_1c_2}$ and let $\overline{pc_2}$ be the straight segment
which is vertical to $\overline{c_1c_2}$ and $p$ is on the perimeter of the
unit disc with center $c_2$ and with direction inside of the convex hull. Next consider the point $c$ on segment
$\overline{pc_2}$ which has distance $\frac{1}{2m} - \epsilon $ from $c_2$. Then point $\mu$ is the
intersection of the straight segment $\overline{c_1c}$ and the perimeter of the
unit disc with center $c_2$. See Figure~\ref{fig:movetopoint} 
for an example.\vspace{3pt}
\hrule\vspace{1em}

\begin{figure}[t]
\centering
\includegraphics[width=3in]{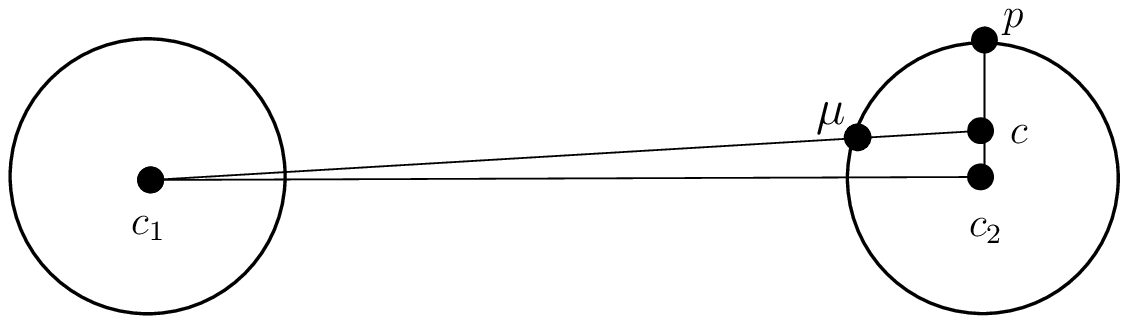}
\caption{Example of a point $\mu$.}
\label{fig:movetopoint}
\end{figure}

It is not difficult to see that Function {\tt Move-to-Point} involves simple
geometric calculations.    

\paragraph{Insight:} This function is called by robot $r$ with center $c_1$.
Point $c_2$ corresponds to the center of a robot that $r$ wants to touch
(i.e., the unit discs representing the robots become tangent). For this purpose
robot $r$ must move towards the other robot in such a way that their circles
become tangent at point $\mu$. As we will see later, this function is called with $m=n$, $n$ being 
the number of robots in the system. Intuitively, distance $\frac{1}{2n} - \epsilon$ is used to
aid robot $r$ to remain visible by other robots (that is, it will not
be hidden by the robot with center $c_2$.)

\subsection{Function {\tt Find-Points}}

Function {\tt Find-Points} solves the following algorithmic problem:\\

\hrule\vspace{3pt}
\noindent {\sc Find Points}\vspace{3pt}
\hrule\vspace{5pt}
\noindent {\bf Input:} Given a Convex Hull, let points $c_1,c_2,\ldots,c_m$ be the points that
are on the convex hull out of the total $n$ points.\vspace{2pt}\\
{\bf Output:} A set of $k<m$ points $p_1,\ldots,p_k$ that a unit disc with center $p_i,~1\leq i\leq k$,
can be placed on the convex hull without causing the convex hull to change. (It is possible
that $k=0$.)\vspace{3pt}
\hrule\vspace{.7em}
 
We now give the details of  function {\tt Find-Points}.\vspace{.5em}

\hrule\vspace{3pt}
\noindent Function: {\tt Find-Points}\vspace{3pt}
\hrule\vspace{5pt}
\noindent {\bf Set} $Points = \emptyset$;\\
Consider the points on the convex hull with a clockwise ordering;\\
{\bf For} each pair $c_l,~c_r$ of neighboring points on the convex hull {\bf do}:\vspace{-.7em}
\begin{itemize}
		\item [] {\bf If} the length of line $\overline{c_lc_r}$ is greater than or equal to 2, 
		{\bf 		then}\vspace{-.6em}
		\begin{itemize}
				\item [] Let $\mu$ be the center of $\overline{c_lc_r}$;\vspace{-.3em}
				\item [] Draw a vertical line on $\overline{c_lc_r}$, that starts from point 
				$\mu$ and moves outside the convex, until distance $\frac{1}{n}$ from $\mu$. Let $p$ be the ending
				point;\vspace{-.3em}
				\item [] Let $c_{l-1}$ be the left neighbor of $c_l$ and $c_{r+1}$ be the right neighbor 
				of $c_r$ {\small\em (mods are omitted for simplicity of notation)};\vspace{-.3em}
				\item [] Consider the straight segment that starts from $c_{l-1}$ and goes through $c_l$
				and the straight segment that starts from $c_{r+i}$ and goes through $c_r$. 
				Let $t$ be the point where the two segments intersect;\vspace{-.3em}  
				\item [] Consider the unit disc formed with center $p$;\vspace{-.3em}
				\begin{itemize}
						\item [] {\bf If} {\em no} point of this unit disc is above {\bf or} p has distance $\frac{1}{n}$ or more from
						 the line segments $\overline{c_{l-1}t}$ and $\overline{tc_{r+1}}$, {\bf then} 
						 $Points = Points \cup \{p\}$; (See Figure~\ref{fig:findpoints} 
						 for an example.)\vspace{-.7em}
				\end{itemize}
	\end{itemize}
\end{itemize}
{\bf Return} $Points$;
\vspace{3pt}
\hrule%\vspace{1em}						 

\paragraph{Insight:} This function is called by a robot that is not on the convex hull
and wishes to see whether there is at least one point that it could move and get on 
the convex hull without causing the convex hull to change. The number of the input points, $m$ is smaller than $n$ i.e. $m<n$, because if this function is called, it means that at least one robot is not on the convex hull.
Given a line segment 
$\overline{c_lc_r}$ of length at least 2, a simple solution would be for the robot to
move in the middle of this line; however, that would cause robots $r_l$ and $r_r$ not
to be visible to each other, which is another property we wish to have (all robots
on the convex hull must be able to see each other). Therefore, we check whether the robot
could be placed at some vertical distance away from the middle, so that $r_l$ and $r_r$
can still see other, but at the same time the tangents on the convex hull are not affected
(in which case it would cause the convex hull to change).  

\begin{lemma}
\label{FindPointsLemma}
If a unit disk is placed on $\CH(V_i)$ with center a point that was returned by Function {\tt FindPoints}, it will not cause $\CH(V_i)$ to change.
\end{lemma}

\begin{proof}
The correctness follows by close investigation of the code of  Function {\tt FindPoints}. If a unit disk moves to one of the points returned, it will not cause $\CH(V_i)$ to change provided that other unit disks do not move

\end{proof}

\begin{lemma}
\label{CHlengthLemma}
Given a convex hull, for any two adjacent unit disks with centers $c_l$ and $c_r$ on the convex hull there exists a minimum distance between $c_l$, $c_r$ for which Function {\tt Find-Points} would return a point between $c_l$, $c_r$. We refer to this distance as the {\em safe} distance.
\end{lemma}

\begin{proof}
Consider that given a number of points, a convex hull always exists. Consider four neighbor points on a convex hull, as shown on Figure~\ref{fig:findpoints}, without loss of generality. In order for a unit disk with center $p$ to be on the convex hull and not cause the current convex hull to change, the distance between $\mu$ , the middle point of $\overline{c_lc_r}$ and $p$ must be at least $\frac{1}{n}$. Note that $p$ is outside of the current convex hull. Additionally, consider $q$ the point on the line segment $\overline{pc_{r+1}}$, where a vertical line to $c_r$ starts from line segment $\overline{pc_{r+1}}$ with direction to the inside of the convex hull.Then $d(q,c_{r})$ must be equal with at least $\frac{1}{n}$, where $r$ is the point that $\overline{pc_{r+1}}$ is tangent with $\overline{\mu c_{r}}$. Angle $\widehat{pr\mu}$ is equal with angle $\widehat{c_r r q}$

We need to calculate the distance between $c_l$ and $c_r$ which will give as the safe distance. The distance between $c_l$ to $\mu$ must be equal with the distance between $c_r$ to $\mu$. We need to calculate both $d(\mu, c_r)$ and $d(\mu ,c_l)$, find the biggest and double it, in order to find the safe distance.
First we must calculate the necessary distance between $\mu$ and $c_r$.

Observe that $d(\mu, c_r)= d(\mu, r) + d(r, c_r)$. Firstly we calculate $d(\mu, r)$
We have that $tan(\widehat{pr\mu} )= \frac{\frac{1}{n}}{d(\mu, r)}$, hence 
$d(\mu, r)= \frac{1}{n \cdot tan(\widehat{pr\mu} )}$
Secondly we calculate $ d(r,c_r)$, we have that $sin(\widehat{c_r r q})= \frac{\frac{1}{n}}{d(r,c_r)}$, hence
$d(r, c_r)= \frac{1}{n \cdot sin((\widehat{c_r r q}) )} = \frac{1}{n \cdot sin((\widehat{pr \mu}) )}$
Now we are ready to calculate $d(\mu, c_r)$, it follows that $d(\mu, c_r)= \frac{1}{n \cdot tan(\widehat{pr\mu} )} + \frac{1}{n \cdot sin((\widehat{pr \mu}) )}$

This is the minimum distance that $\overline{\mu c_r}$ must be. We do the same as above with $\overline{\mu c_l}$ and choose the biggest distance between the two, double it and set it as safe distance.

\end{proof}

\begin{figure}[t]
\centering
\includegraphics[width=3.5in]{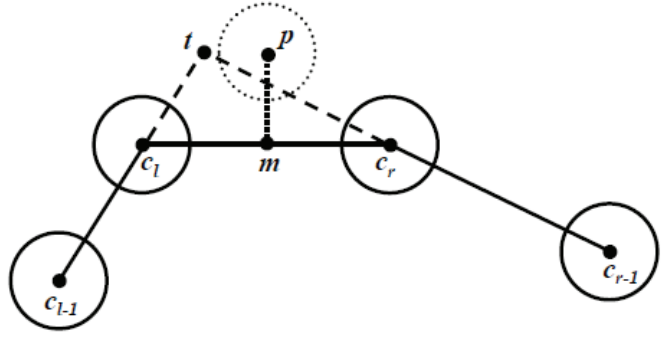}\vspace{-1em}
\caption{An example where point $p$ is not valid and hence
it will not be included in set $Points$.}
\label{fig:findpoints}
\end{figure}

\subsection{Function {\tt Connected-Components}}
\label{subsec:RC}

Consider a set of $m$ unit discs on the plane. 
A {\em connected component} of this set consists all 
unit discs that are connected (between any two
points of any two unit discs there exists a polygonal
line each of whose points belong to some unit disc). 
In a connected component there can be up to two empty spaces 
of distance less or equal to $1/2m$. If there are more
than two such spaces, then this component 
%{\bf Chrysovalandis: the configuration??? ennoo to kommati me ta syndedemena robot} 
is considered as another connected component. Note that a given set of unit discs
may contain many connected components and only one in the case 
that all unit discs are connected.\vspace{.5em}

\paragraph{High level idea:}
A component part consists of
unit disks that are tangent and there exists a polygonal line each of whose points belong to a robot.
A component can have one, two or three component parts with maximum space between adjacent parts $\frac{1}{2n}$.
This is a simple problem but with many possible cases that can be seen on the code of this function. We now proceed with the full code of the function.

Function {\tt Connected-Components} solves the following algorithmic problem:\\

\hrule\vspace{3pt}
\noindent {\sc Connected Components}\vspace{3pt}
\hrule\vspace{5pt}
\noindent {\bf Input:} A set of $m$ points $c_1,c_2,\ldots,c_m$ and an additional 
point $c$.\vspace{2pt}\\
{\bf Output:} A set of pairs of the form $\langle(c_l,c_r),~k\rangle$. 
Each pair $(c_l,c_r)$ represents a connected
component of unit discs, where $c_l$ is the center of the leftmost unit disc
and $c_r$ the center of the rightmost unit disc in the component. $k$ is the
number of unit discs contained in this component (including $c_l$ and $c_r$).\vspace{3pt}
\hrule\vspace{.7em}
 
We now give the details of  function {\tt Connected-Components}. The correctness of 
the function (that is, the proof that it correctly solves the above problem) follows
by close investigation of the code of the funciton. 
We use the notation $\langle (c_{rx},c_{ly}),k_{rx,ly}\rangle$, to denote
the component where $c_{rx}$ is the center of the first unit disc on the right of $spacex$, 
	$c_{ly}$ is the center of the last unit disc on the left of $spacey$, and $k_{rx,ly}$ is
	the number of unit discs between (and including) $c_{rx}$ and $c_{ly}$; $x,y,k_{rx,ly}$
	are positive integers.\vspace{.7em}

\hrule\vspace{3pt}
\noindent Function: {\tt Connected-Components}\vspace{3pt}
\hrule\vspace{5pt}
\noindent {\bf Set} $initial = c$ and $Components = \emptyset$;\\
Starting from $initial$, move to the {\tt right} along connected unit discs until 
a space is reached. Call this space, $space0$;\vspace{-.5em}
\begin{enumerate} %A-B
%A
\item {\bf If} the length of $space0$ is less than or equal to $\frac{1}{2m}$ {\bf then}
continue moving to the {\tt right} along connected unit discs until another
space is reached. Call this space, $space1$;\vspace{-.3em} 
	\begin{enumerate} %1-2
	%1
	\item {\bf If} the length of $space1$ is less than or equal to $\frac{1}{2m}$ {\bf then}
	continue moving to the {\tt right} along connected unit discs until another
	space is reached. Call this space, $space2$;\vspace{-.2em}
		\begin{enumerate} %I-II
		%I
		\item {\bf If} the length of $space2$ is less than or equal to $\frac{1}{2m}$ {\bf then}
		continue moving to the {\tt right} along connected unit discs until another
		space is reached. Call this space, $space3$;\\ %I
		{\bf Set} $Components = Components \cup \langle (c_{r2},c_{l3}),k_{r2,l3}\rangle \cup 
		\langle (c_{r1},c_{l2}),k_{r1,l2}\rangle$;\\
		From $initial$, move to the {\tt left} along connected unit discs until a space is
		reached. Call this space, $space4$;\\
  	{\bf Set} $Components = Components \cup \langle (c_{r4},c_{l0}),k_{r4,l0}\rangle
  	\cup \langle (c_{r0},c_{l1}),k_{r0,l1}\rangle$;\\
		{\bf Set} $initial = c_{r3}$ (center of first unit disc on the right of $space3$).\\
		{\bf If} the unit disc with center $c$ (the input's additional point) is included in one of the 		newly included components {\tt and} this is
		not the first iteration of the procedure, {\bf then} {\bf remove} multiplicities and 
		{\bf terminate};\\
		{\bf Else} {\bf repeat} procedure with (new) $initial$.\vspace{-.1em}
		%II
		\item  {\bf If} the length of $space2$ is greater than $\frac{1}{2m}$ {\bf then}
		from $initial$, move to the {\tt left} along connected unit discs until a
		space is reached. Call this space, $space3$;\\ %II
		\begin{enumerate} %a-b
		%a
			\item{\bf If} the length of $space3$ is greater than $\frac{1}{2m}$ {\bf then}
			{\bf Set} $Components = Components \cup \langle (c_{r3},c_{l2}),k_{r3,l2}\rangle$;\\
			{\bf Set} $initial = c_{r2}$.\\
			{\bf If} the unit disc with center $c$ (the input's additional point) is included in the newly 			included component {\tt and} this is
			not the first iteration of the procedure, {\bf then} {\bf remove} multiplicities and 
			{\bf terminate};\\
			{\bf Else} {\bf repeat} procedure with (new) $initial$.\vspace{-.2em}
		%b
			\item{\bf If} the length of $space3$ is less than or equal to $\frac{1}{2m}$ {\bf then}
			{\bf Set} $Components = Components \cup \langle (c_{r0},c_{l1}),k_{r0,l1}\rangle \cup \langle (c_{r1},c_{l2}),k_{r1,l2}\rangle$;\\
			continue moving to the left along connected unit disks until another space is reached. Call this space, $space 4$;
			{\bf Set} $Components = Components \cup \langle (c_{r4},c_{l3}),k_{r4,l3}\rangle \cup \langle (c_{r3},c_{l0}),k_{r3,l0}\rangle$;\\
			{\bf Set} $initial = c_{r2}$.\\
			{\bf If} the unit disc with center $c$ (the input's additional point) is included in the newly 			included component {\tt and} this is
			not the first iteration of the procedure, {\bf then} {\bf remove} multiplicities and 
			{\bf terminate};\\
			{\bf Else} {\bf repeat} procedure with (new) $initial$.\vspace{-.2em}
		\end{enumerate}      
		\end{enumerate}
	%2
	\item {\bf If} the length of $space1$ is greater than $\frac{1}{2m}$ {\bf then}
		from $initial$, move to the {\tt left} along connected unit discs until a
		space is reached. Call this space, $space2$;\vspace{-.2em} 
			\begin{enumerate} %I-II
			% I
			\item {\bf If} the length of $space2$ is greater than $\frac{1}{2m}$ {\bf then}
			{\bf Set} $Components = Components \cup \langle (c_{r2},c_{l1}),k_{r2,l1}\rangle$;\\
			{\bf Set} $initial = c_{r1}$.\\
			{\bf If} the unit disc with center $c$ (the input's additional point) is included 
			in the newly included component {\tt and} this is
			not the first iteration of the procedure, {\bf then} {\bf remove} multiplicities and 
			{\bf terminate};\\
			{\bf Else} {\bf repeat} procedure with (new) $initial$.\vspace{-.2em}
			% II
			\item {\bf If} the length of $space2$ is less than or equal to $\frac{1}{2m}$ {\bf then}
			continue moving to the {\tt left} along connected unit discs until another
			space is reached. Call this space, $space3$;\vspace{-.2em}
				\begin{enumerate} %a-b
				%a
				\item {\bf If} the length of $space3$ is greater than $\frac{1}{2m}$ {\bf then}
				{\bf Set} $Components = Components \cup \langle (c_{r3},c_{l1}),k_{r3,l1}\rangle$;\\
				{\bf Set} $initial = c_{r1}$.\\
				{\bf If} the unit disc with center $c$ (the input's additional point) is included 
				in the newly included component {\tt and} this is
				not the first iteration of the procedure, {\bf then} {\bf remove} multiplicities and 
				{\bf terminate};\\
				{\bf Else} {\bf repeat} procedure with (new) $initial$.\vspace{-.2em}
				%b
				\item {\bf If} the length of $space3$ is less than or equal to $\frac{1}{2m}$ {\bf then}
				continue moving to the {\tt left} along connected unit discs until another
				space is reached. Call this space, $space4$;\\
				{\bf Set} $Components = Components~\cup$ 
				$\langle (c_{r0},c_{l1}),k_{r0,l1}\rangle~\cup$ 
				$\langle (c_{r2},c_{l0}),k_{r2,l0}\rangle~\cup$
				$\langle (c_{r3},c_{l2}),k_{r3,l2}\rangle~\cup$
				$\langle (c_{r4},c_{l3}),k_{r4,l3}\rangle$;\\
				{\bf Set} $initial = c_{r1}$.\\
				{\bf If} the unit disc with center $c$ (the input's additional point) is included 
				in the newly included component {\tt and} this is
				not the first iteration of the procedure, {\bf then} {\bf remove} multiplicities and 
				{\bf terminate};\\
				{\bf Else} {\bf repeat} procedure with (new) $initial$.\vspace{-.2em} 
				\end{enumerate}       
			\end{enumerate}
	\end{enumerate}
%B
\item {\bf If} the length of $space0$ is greater than $\frac{1}{2m}$ {\bf then}
		from $initial$, move to the {\tt left} along connected unit discs until a
		space is reached. Call this space, $space1$;\vspace{-.2em} 
			\begin{enumerate} %1-2
			% 1
			\item {\bf If} the length of $space1$ is greater than $\frac{1}{2m}$ {\bf then}
			{\bf Set} $Components = Components \cup \langle (c_{r1},c_{l0}),k_{r1,l0}\rangle$;\\
			{\bf Set} $initial = c_{r0}$.\\
			{\bf If} the unit disc with center $c$ (the input's additional point) is included 
			in the newly included component {\tt and} this is
			not the first iteration of the procedure, {\bf then} {\bf remove} multiplicities and 
			{\bf terminate};\\
			{\bf Else} {\bf repeat} procedure with (new) $initial$.\vspace{-.2em}
			% 2
			\item {\bf If} the length of $space1$ is less than or equal to $\frac{1}{2m}$ {\bf then}
			continue moving to the {\tt left} along connected unit discs until another
			space is reached. Call this space, $space2$;\vspace{-.2em}
				\begin{enumerate} %I-II
				%I
				\item {\bf If} the length of $space2$ is greater than $\frac{1}{2m}$ {\bf then}
				{\bf Set} $Components = Components \cup \langle (c_{r2},c_{l0}),k_{r2,l0}\rangle$;\\
				{\bf Set} $initial = c_{r0}$.\\
				{\bf If} the unit disc with center $c$ (the input's additional point) is included 
				in the newly included component {\tt and} this is
				not the first iteration of the procedure, {\bf then} {\bf remove} multiplicities and 
				{\bf terminate};\\
				{\bf Else} {\bf repeat} procedure with (new) $initial$.\vspace{-.2em}
				%II
				\item {\bf If} the length of $space2$ is less than or equal to $\frac{1}{2m}$ {\bf then}
				continue moving to the {\tt left} along connected unit discs until another
				space is reached. Call this space, $space3$;\vspace{-.2em}
					\begin{enumerate} %a-b
					%a
					\item {\bf If} the length of $space3$ is greater than $\frac{1}{2m}$ {\bf then}
					{\bf Set} $Components = Components \cup \langle (c_{r3},c_{l0}),k_{r3,l0}\rangle$;\\
					{\bf Set} $initial = c_{r0}$.\\
					{\bf If} the unit disc with center $c$ (the input's additional point) is included 
					in the newly included component {\tt and} this is
					not the first iteration of the procedure, {\bf then} {\bf remove} multiplicities and 
					{\bf terminate};\\
					{\bf Else} {\bf repeat} procedure with (new) $initial$.\vspace{-.2em}
				  %b  				
					\item {\bf If} the length of $space3$ is less than or equal to $\frac{1}{2m}$ {\bf then}
					continue moving to the {\tt left} along connected unit discs until another
					space is reached. Call this space, $space4$;\\	
					{\bf Set} $Components = Components~\cup$ 
					$\langle (c_{r1},c_{l0}),k_{r1,l0}\rangle~\cup$ 
					$\langle (c_{r2},c_{l1}),k_{r2,l1}\rangle~\cup$
					$\langle (c_{r3},c_{l2}),k_{r3,l2}\rangle~\cup$
					$\langle (c_{r4},c_{l3}),k_{r4,l3}\rangle$;\\
					{\bf Set} $initial = c_{r0}$.\\
					{\bf If} the unit disc with center $c$ (the input's additional point) is included 
					in the newly included component {\tt and} this is
					not the first iteration of the procedure, {\bf then} {\bf remove} multiplicities and 
					{\bf terminate};\\
					{\bf Else} {\bf repeat} procedure with (new) $initial$.\vspace{-.3em}
					\end{enumerate} 
				\end{enumerate}       
			\end{enumerate}
\end{enumerate}		
\vspace{3pt}
\hrule%\vspace{1em}

\paragraph{Insight:} This function is called by a robot $r$ with center $c$. The
$m$ points are the centers of the robots that robot $r$ can see (its local view)
in the current configuration. As we will see later, this function is called
when the robot can see all other robots, i.e., $m=n$. The robot wishes to find
the connected components formed in the current configuration. Intuitively, 
we can include two spaces of length $1/2n$ in a configuration, since if all the
robots can see each other, 
%are in symmetry ({\bf Chrysovalandis: which if they are, all robots can understand it} since all can see all other), 
then the robots can move taking steps of length $1/2n$ until they meet.

\subsection{Function {\tt How-Much-Distance}}

Function {\tt How-Much-Distance} solves the following algorithmic problem:\\

\hrule\vspace{3pt}
\noindent {\sc How Much Distance}\vspace{3pt}
\hrule\vspace{5pt}
\noindent {\bf Input:} A set of $m$ points $c_1,c_2,\ldots,c_m$ and an additional 
point $c$.\vspace{2pt}\\
{\bf Output:} One of the numbers 1,2 or 3. Consider the connected components
formed by the unit discs with centers $c_1,c_2,\ldots,c_m$. If the unit disc with center
$c$ is the rightmost (the straight direction is considered to be the inside of the convex hull) element of the component
that has the smallest (space-wise) distance between the components, then the answer is
1. If all components have the same distance, then the answer is 2. 
Otherwise the answer is 3.\vspace{3pt}
\hrule\vspace{1em}

Function {\tt How-Much-Distance} calls function 
{\tt Connected-Components} (Section~\ref{subsec:RC})
to get the connected components formed by the unit discs with centers $c_1,c_2,\ldots,c_m$. 
Then it checks the distances between the components and returns 1,2 or 3 accordingly.\vspace{-.5em} 

\paragraph{Insight:} This function is called by robot $r$ with center $c$. The robot
wants to check whether it is the rightmost robot in the connected component
with the smallest distance among the components formed by the robots in its local
view.

\subsection{Function {\tt In-Largest-Component}}

Function {\tt In-Largest-Component} solves the following algorithmic problem:\\

\hrule\vspace{3pt}
\noindent {\sc In Largest Component}\vspace{3pt}
\hrule\vspace{5pt}
\noindent {\bf Input:} A set of $m$ points $c_1,c_2,\ldots,c_m$ and an additional 
point $c$.\vspace{2pt}\\
{\bf Output:} One of the numbers 1,2 or 3. Consider the connected components
formed by the unit discs with centers $c_1,c_2,\ldots,c_m$. If the unit disc with center
$c$ belongs in the largest component (wrt the number of discs), then the answer is 1;
if all the components are larger than the one it belongs, then the answer is 2. Otherwise
the answer is 3.\vspace{3pt}
\hrule\vspace{1em}

Function {\tt In-Largest-Component} calls function {\tt Connected-Components} (Section~\ref{subsec:RC})
to get the connected components formed by the unit discs with centers $c_1,c_2,\ldots,c_m$.
Then it checks the sizes of the components and where the unit disc with center $c$ belongs to,
and returns 1,2 or 3 accordingly. 

\paragraph{Insight:} This function is called by robot $r$ with center $c$. The robot
wants to check whether it belongs in the largest component among the components
formed by the robots in its local view. 

\subsection{Function {\tt In-Smallest-Component}}

Function {\tt In-Smallest-Component} solves the following algorithmic problem:\\

\hrule\vspace{3pt}
\noindent {\sc In Smallest Component}\vspace{3pt}
\hrule\vspace{5pt}
\noindent {\bf Input:} A set of $m$ points $c_1,c_2,\ldots,c_m$ and an additional 
point $c$.\vspace{2pt}\\
{\bf Output:} One of the numbers 1,2 or 3. Consider the connected components
formed by the unit discs with centers $c_1,c_2,\ldots,c_m$. If the unit disc with center
$c$ belongs in the smallest component (wrt the number of discs), then the answer is 1;
if all the components are smaller than the one it belongs, then the answer is 2. Otherwise
the answer is 3.\vspace{3pt}
\hrule\vspace{1em}

Function {\tt In-Smallest-Component} calls function {\tt Connected-Components} (Section~\ref{subsec:RC})
to get the connected components formed by the unit discs with centers $c_1,c_2,\ldots,c_m$.
Then it checks the sizes of the components and where the unit disc with center $c$ belongs to,
and returns 1,2 or 3 accordingly. 

\paragraph{Insight:} This function is called by robot $r$ with center $c$. The robot
wants to check whether it belongs in the smallest component among the components
formed by the robots in its local view.

\subsection{Function {\tt In-Straight-Line-2}}

Function {\tt  In-Straight-Line-2} solves the following decision problem:\\

\hrule\vspace{3pt}
\noindent {\sc  In Straight Line 2}\vspace{3pt}
\hrule\vspace{5pt}
\noindent {\bf Input:} A set of $3$ points $c_l,c_m$ and $c_r$.\vspace{2pt}\\
{\bf Output:} $YES$, if the three points are on the same line. Otherwise, $NO$.\vspace{3pt}
\hrule\vspace{1em} 

Function {\tt In-Straight-Line-2} involves simple geometric calculations to check if the three input points are on the same straight line.

\paragraph{Insight:} This function is called by robot $r_m$ with center $c_m$. The robot $r_m$
wants to check whether it is on the same line with its left nearest neighbor robot on the convex hull, $r_l$ with center $c_l$, and with its right nearest neighbor robot, $r_r$ with center $c_r$. 

%%%%%%%%%%%%%%%%%%%%%%%%%%%%% Local Algorithm %%%%%%%%%%%%%%%%%%%%%%%%%%%%%%%%%%%%

\section{Local Algorithm for Compute}
\label{sec:algorithm}

In this section we present the algorithm that each robot runs locally
while in state {\bf compute}. This algorithm takes as an input
the view of the robot (obtained in state {\bf look}) and calculates
the position in the plane the robot should move next (in state {\bf move}).
%in order to solve the gathering problem. 
%Each robot runs
%an instance of the algorithm locally. These instances constitute
%together a distributed algorithm.

In Section~\ref{subsec:algstates} we overview the states of the algorithm and
in Section~\ref{subsec:alg-detail} we give a detail description of the algorithm
along with key observations/properties. 
%Its correctness is shown in Section~\ref{subsec:correctness}.

\subsection{States of the Algorithm}
\label{subsec:algstates}

Once a robot $r_i$ is in state {\bf Compute} it starts executing the local
algorithm $A_i$. Recall that $V_i$ denotes robot's $r_i$ local view, that is, 
the set of robots that are visible to robot $r_i$ upon entering state {\bf Compute}.
The algorithm consists of 17 states. These states are algorithmic states within state {\bf Compute} and we refer to them using the notation
{\bf Compute.$\langle$algorithm-state-name$\rangle$}. We now overview these states with respect to a robot $r_i$.

\begin{enumerate}
\item {\bf Compute.Start}:
\begin{itemize}
\item This is the initial state of the algorithm run by robot $r_i$. 
\end{itemize}
\item {\bf Compute.OnConvexHull}:
\begin{itemize}
\item Robot $r_i$ is on the convex hull formed by robots in its local view.
\end{itemize}
\item \label{st:AllOnCH} {\bf Compute.AllOnConvexHull}:
\begin{itemize}
\item Robot $r_i$ is on the convex hull.
\item Robot $r_i$ can see all other $n-1$ robots (that is, it has full visibility).
\item All other $n-1$ robots are on the convex hull and have full visibility.
\end{itemize} 
\item {\bf Compute.Connected}:
\begin{itemize}
\item Same conditions as in state~\ref{st:AllOnCH}
\item Robot $r_i$ sees that all robots are connected.
\end{itemize}
\item {\bf Compute.NotConnected}:
\begin{itemize}
\item Same conditions as in state~\ref{st:AllOnCH}
\item Robot $r_i$ sees that not all robots are connected.
\end{itemize}
\item \label{st:NAllOnCH} {\bf Compute.NotAllOnConvexHull}:
\begin{itemize}
\item Robot $r_i$ is on the convex hull.
\item Robot $r_i$ cannot see all other $n-1$ robots or at least one robot is not on the convex hull or all robots are on the convex hull, but there is at least one that does not have full visibility.
\end{itemize}
\item \label{st:NSL} {\bf Compute.NotOnStraightLine}:
\begin{itemize}
\item Same conditions as in state~\ref{st:NAllOnCH}.
\item There are no two other robots on the same line with robot $r_i$ (all three are on the convex hull).
\end{itemize}
\item \label{st:SForMore} {\bf Compute.SpaceForMore}:
\begin{itemize}
\item Same conditions as in state~\ref{st:NSL}.
\item Robot $r_i$ sees that there is space on the convex hull for another robot. That is,
 there exist two neighboring robots on the convex hull that their distance is at least 2 (recall that robots are unit discs).  
\end{itemize}

\item {\bf Compute.NoSpaceForMore}:
\begin{itemize}
\item Same conditions as in state~\ref{st:NSL}.
\item Robot $r_i$ sees that there is no space on the convex hull for another robot; all neighboring robots on the convex hull have distance less than 2. 
\end{itemize}
\item \label{st:OnSL} {\bf Compute.OnStraightLine}:
\begin{itemize}
\item Same conditions as in state~\ref{st:NAllOnCH}.
\item There are at least two other robots on the same line with robot $r_i$ (all three are on the convex hull).
\end{itemize}
\item {\bf Compute.SeeOneRobot}:
\begin{itemize}
\item Same conditions as in state~\ref{st:OnSL}.
\item Robot $r_i$ can see only one robot on the line.
\end{itemize}
\item {\bf Compute.SeeTwoRobot}:
\begin{itemize}
\item Same conditions as in state~\ref{st:OnSL}.
\item Robot $r_i$ can see two robots on the line; this implies
that robot $r_i$ is between these two robots.
\end{itemize}
\item \label{st:NotOCH} {\bf Compute.NotOnConvexHull}:
\begin{itemize}
\item Robot $r_i$ is enclosed in the convex hull formed by robots in its local view.
\end{itemize}
\item {\bf Compute.IsTouching}:
\begin{itemize}
\item Same conditions as in state~\ref{st:NotOCH}.
\item Robot $r_i$ is touching another robot (the unit discs representing the robots 
are tangent).
\end{itemize}
\item \label{st:NotTouching} {\bf Compute.NotTouching}:
\begin{itemize}
\item Same conditions as in state~\ref{st:NotOCH}.
\item Robot $r_i$ does not touch any other robot.
\end{itemize}
\item {\bf Compute.ToChange}:
\begin{itemize}
\item Same conditions as in state~\ref{st:NotTouching}.
\item If robot $r_i$ moves as calculated by the algorithm, then it will cause the convex hull to change, and this cannot be avoided. 
\end{itemize}
\item {\bf Compute.NotChange}:
\begin{itemize}
\item Same conditions as in state~\ref{st:NotTouching}.
\item If robot $r_i$ moves as calculated by the algorithm, then there is a way to avoid changing the convex hull. 
\end{itemize}
\end{enumerate}

Figure~\ref{fig:AlgStates} depicts all possible states and transitions of the
algorithm run by robot $r_i$. (For better readability the prefix {\bf Compute} is voided.) 
The states that have no transition to another state are terminal, and
they output the position that the robot will move next (and the robot exits state {\bf Compute}
and enters state {\bf Move}). State {\bf Compute.Connected} outputs the special point~$\bot$
which causes robot $r_i$ to exit state {\bf Compute} and enter state {\bf Terminate} (the
robot takes no further steps).
%gathering problem has been solved).   

\begin{figure}[t]
\includegraphics[width=6.2in]{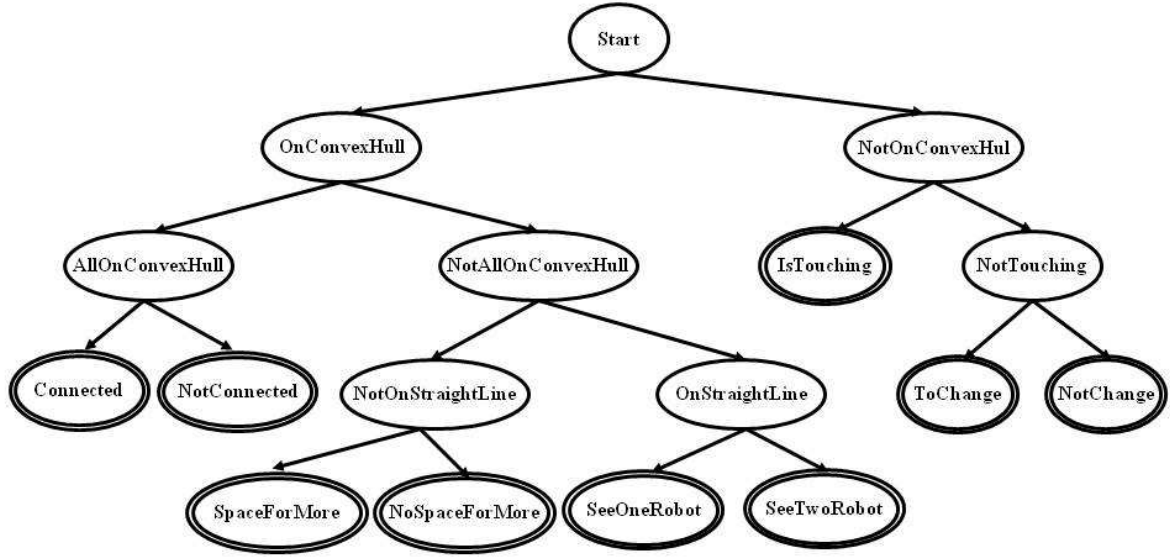}\vspace{-10em} %for dvi/ps
\caption{All possible states and transitions of the algorithm
run by robot $r_i$.}
\label{fig:AlgStates} 
\end{figure}

\subsection{Description of the Algorithm}
\label{subsec:alg-detail}

The algorithm consists of $17$ procedures, each treating one of
the possible algorithmic states. In particular, once the algorithm is in
a state {\bf Compute.$\langle$algorithm-state-name$\rangle$} it
runs the corresponding procedure {\tt algorithm-state-name} 
that either implements a state transition or outputs a point on
the plane the robot should move at (in the next state {\bf Move}); 
it implements a state transition if
it is in a non-terminal state and outputs a point otherwise.  In a nutshell,
the algorithm can be expressed as follows:\\

\hrule\vspace{3pt}
\noindent {\sc Local Algorithm}\vspace{3pt}
\hrule\vspace{5pt}
\noindent {\bf if} state= {\bf Compute.$\langle$algorithm-state-name$\rangle$}
{\bf then} run procedure {\tt algorithm-state-name}.\vspace{3pt}
\hrule\vspace{1em}

We proceed to describe the procedures. The procedures are given
with respect to a robot $r_i$.

\subsubsection{Procedure {\tt Start}}

\hrule\vspace{3pt}
\noindent Procedure: {\tt Start}\vspace{3pt}
\hrule\vspace{5pt}
\noindent {\em Precondition}: state= {\bf Compute.Start}\\
\noindent {\em Effect}:\vspace{-.8em}
\begin{itemize}
\item {\bf Call} function {\tt On-Convex-Hull} with inputs, the local view
$V_i$ and $c_i$, the center of robot $r_i$.\vspace{-.7em}
\item {\bf If} function {\tt On-Convex-Hull} returns {\sc yes} {\bf then} 
state:= {\bf Compute.OnConvexHull}\vspace{-.7em}
\item [] {\bf Else} state:= {\bf Compute.NotOnConvexHull} 
\end{itemize}\vspace{-.4em} 
%\vspace{3pt}
\hrule\vspace{1em}

Recall from Section~\ref{subsec:FOnCH} that function {\tt On-Convex-Hull} also
returns the set of points that are on the convex hull based on robot's $r_i$ view $V_i$. We will
be denoting this set as $\CH(V_i)$. From this point onwards, robot $r_i$
carries the knowledge of $\CH(V_i)$ in the various algorithmic states, while in state
{\bf Compute} (this knowledge is lost once it exists this state). 

\begin{lemma}
\label{StartLemma} 
{\tt Start}{\bf(Compute.$\langle$Start$\rangle$)} = {\bf Compute.$\langle$OnConvexHull$\rangle$} iff $c_i \in \CH(V_i)$. 
\end{lemma}

\begin{proof}
Procedure {\tt Start} uses function {\tt On-Convex-Hull}, which uses Grahams algorithm\cite{Graham72} that was proven to be correct.
\hfill$\blacksquare$
\end{proof}

%\subsubsection{Procedure {\tt OnConvexHull}}
%
%\hrule\vspace{3pt}
%\noindent Procedure: {\tt OnConvexHull}\vspace{3pt}
%\hrule\vspace{5pt}
%\noindent {\em Precondition}: state= {\bf Compute.OnConvexHull}\\
%\noindent {\em Effect}:\vspace{-.8em}
%\begin{itemize}
%\item {\bf If} $|V_i| = n$ {\tt and} $|CH|=n$ {\bf then} 
%%for each $c_j\in V_i-\{c_i\}$ {\bf call} 
%%function {\tt On-Convex-Hull} with inputs, the local view
%%$V_i$ and $c_j$.\vspace{-.7em}
%%\begin{itemize}
%%\item {\bf If} function {\tt On-Convex-Hull} returns {\sc yes} {\em for all} $c_j$
%%{\bf then} 
%state:= {\bf Compute.AllOnConvexHull}\vspace{-.7em}
%\item[] {\bf Else} state:= {\bf Compute.NotAllOnConvexHull}%\vspace{-.7em}
%%\end{itemize} 
%%\item {\bf Else} state:= {\bf Compute.NotAllOnConvexHull} 
%\end{itemize}\vspace{-.4em} 
%%\vspace{3pt}
%\hrule%\vspace{1em}

\subsubsection{Procedure {\tt OnConvexHull}}

\hrule\vspace{3pt}
\noindent Procedure: {\tt OnConvexHull}\vspace{3pt}
\hrule\vspace{5pt}
\noindent {\em Precondition}: state= {\bf Compute.OnConvexHull}\\
\noindent {\em Effect}:\vspace{-.8em}
\begin{itemize}
\item {\bf If} $|V_i| = n$ {\tt and} $|\CH(V_i)|=n$ {\bf then} 
\begin{itemize}
\item for each $r_j$ $\in$ $V_i-\{r_i\}$
\begin{itemize}
\item {\bf Call} function {\tt On-Straight-Line-2} with inputs: the center of the left neighbor of $r_j$, the center of $r_j$ and the center of the right neighbor of $r_j$.
\item {\bf If} function {\tt On-Straight-Line-2} returned $YES$ {\bf then}
state:= {\bf Compute.NotAllOnConvexHull} and {\bf return}
\end{itemize}
\item state:= {\bf Compute.AllOnConvexHull}

\end{itemize}
\item[] {\bf Else} state:= {\bf Compute.NotAllOnConvexHull}%\vspace{-.7em}

\end{itemize}\vspace{-.4em} 
\hrule%\vspace{1em}

\begin{lemma}
\label{OnConvexHullLemma}
{\tt OnConvexHull}{\bf(Compute.$\langle$OnConvexHull$\rangle$)} = {\bf Compute.$\langle$AllOnConvexHull$\rangle$} iff $|V_i| = n$ and
 $|\CH(V_i)|=n$ and all robots have full visibility, according to $V_i$. 
\end{lemma}

\begin{proof}
Based on Lemma ~\ref{StartLemma} $c_i$ $\in$ $\CH(V_i)$ .
Then, there are four possible cases:
\begin{enumerate}
\item  $|V_i| < n$.

In this case the procedure {\tt OnConvexHull} will change state to {\bf Compute.NotAllOnConvexHull} because if $r_i$ can't see all robots, it means that $r_i$ does not have full visibility.
\item  $|V_i| = n$ and $|\CH(V_i)|<n$.

In this case procedure {\tt OnConvexHull} will change state to {\bf Compute.NotAllOnConvexHull} because if at least one robot is not on the convex hull, then the correct state to shift is {\bf Compute.NotAllOnConvexHull}. The procedure {\tt OnConvexHull} can easily determine if all robots are on the convex hull by comparing the number of robots, with the number of robots that belong to $\CH(V_i)$.

\item  $|V_i| = n$ and $|\CH(V_i)| = n$ and at least one robot does not have full visibility.

In this case procedure {\tt OnConvexHull} will  change state to {\bf Compute.NotAllOnConvexHull} because if at least one robot does not have full visibility, the correct state to shift is {\bf Compute.NotAllOnConvexHull}. The procedure {\tt OnConvexHull} can easily determine if all robots have full visibility (according to $V_i$) by checking if all robots are on the convex hull and no three robots are on the same line. From the definition of the convex hull it is clear that no incisions are allowed and hence if no three robots are on the same line, all robots will have full visibility.

\item  $|V_i| = n$ and $|\CH(V_i)| = n$ and all robots have full visibility.

In this case procedure {\tt OnConvexHull} will change state to {\bf Compute.AllOnConvexHull} because if all robots are on the convex hull and have full visibility, the correct state to shift is {\bf Compute.AllOnConvexHull}. The procedure {\tt OnConvexHull} can easily determine if all robots have full visibility (accroding to $V_i$) by checking if all robots are on the convex hull and no three robots are on the same line. From the definition of the convex hull it is clear that no incisions are allowed and hence if no three robots are on the same line, all robots will have full visibility.
\end{enumerate}\vspace{-1.8em}
\hfill$\blacksquare$
\end{proof}

\subsubsection{Procedure {\tt AllOnConvexHull}}

\hrule\vspace{3pt}
\noindent Procedure: {\tt AllOnConvexHull}\vspace{3pt}
\hrule\vspace{5pt}
\noindent {\em Precondition}: state= {\bf Compute.AllOnConvexHull}\\
\noindent {\em Effect}:\vspace{-.8em}
\begin{itemize}
\item {\bf Choose} a $c_j$ from $V_i$. {\bf Set} $Component = \{c_j\}$.\vspace{-.7em}
\item {\bf While} set $Component$ changes {\bf do}\vspace{-.7em}
\begin{itemize}
\item For each $c_j\in Component$ {\bf do}\vspace{-.3em}
\begin{itemize}
\item {\bf Check} for each $c_x\in V_i - Component - \{c_j\}$ whether
its unit disc is tangent with the unit disc of $c_j$. {\bf If} true, {\bf then}  
$Component = Component \cup \{c_x\}.$\vspace{-.7em}
\end{itemize} 
\end{itemize}  
\item {\bf If} $|Component| = n$ {\bf then} state:= {\bf Compute.Connected}\vspace{-.6em}
\item[] {\bf Else} state:= {\bf Compute.NotConnected}\vspace{-.4em}  
\end{itemize}
\hrule

\begin{lemma}
\label{AllOnConvexHullLemma}
{\tt AllOnConvexHull}{\bf (Compute.$\langle$AllOnConvexHull$\rangle$)} = {\bf Compute.$\langle$Connected$\rangle$} iff $V_i$ is a connected configuration. 
\end{lemma}

\begin{proof}
Based on Lemma ~\ref{OnConvexHullLemma}, $|\CH(V_i)|=n$ and all robots have full visibility.
The procedure {\tt AllOnConvexHull} uses simple geometric calculations to calculate the number of the robots that are connected to a random robot $r_j$. If the number of the connected robots is $n$, all robots are connected and therefore, the correct state to shift is {\bf Compute.Connected}.
\hfill$\blacksquare$
\end{proof}

\subsubsection{Procedure {\tt Connected}}
\hrule\vspace{3pt}
\noindent Procedure: {\tt Connected}\vspace{3pt}
\hrule\vspace{5pt}
\noindent {\em Precondition}: state= {\bf Compute.Connected}\\
\noindent {\em Effect}:\vspace{-.8em}
\begin{itemize}
\item {\bf Return} the special point $\bot$.\vspace{-.4em} 
\end{itemize}
\hrule\vspace{.7em}

Once this procedure is executed by robot $r_i$, it enters state {\bf terminate} and does not perform any further steps. 
%and As we will show in the next section, once the robot reaches this
%state, it will never enter the {\bf Compute} state again, and the gathering problem 
%has been solved.

\remove{% ommiting 
\begin{lemma}
\label{ConnectedLemma}
{\tt Connected}{\bf (Compute.$\langle$Connected$\rangle$)} = $\bot$ $iff$ all robots have $full$ $visibility$ and are connected.
\end{lemma}

\begin{proof}
Based on Lemma ~\ref{OnConvexHullLemma}, $|\CH(V_i)|=n$ and all robots have $full$ $visibility$. Based on Lemma ~\ref{AllOnConvexHullLemma} all robots form a $connected$ configuration. Because of these facts, configuration $G$ is $connected$ and hence the algorithm must terminate.
\end{proof}
}

\subsubsection{Procedure {\tt NotConnected}}

\paragraph{High level idea:} The purpose of this procedure is eventually all robots to form a $connected$ configuration. This procedure gives priority to
components with the smallest size firstly and secondly to components that the distance to their neighbor
component on the right is the smallest distance between any two components. The rightmost robot of the
component with the biggest priority moves to the left of its right neighbor component. If all components
have equal priority (i.e. all components have the same size and the distance between any two components is
the same) then robots start to converge. Robots can start moving only if for any three neighbor robots of the
 component, say $r_l,r_m$ and $r_r$ the vertical distance from line $r_l,r_r$ to $r_m$ is equal or more than $\frac{1}{n}$. We proceed with the complete code of the procedure.
\vspace{3pt}

\hrule\vspace{5pt}
\noindent Procedure: {\tt NotConnected}\vspace{3pt}
\hrule\vspace{5pt}
\noindent {\em Precondition}: state= {\bf Compute.NotConnected}\\
\noindent {\em Effect}:\vspace{-.8em}
\begin{itemize}
\item Consider $c_m$, the center of the left neighbor of $r_i$ and $c_l$ the center of the left neighbor of $r_m$.
\begin{itemize}
\item {\bf If} the distance between the line $c_i$,$c_l$ and the point $c_m$ is $<\frac{1}{n}$, {\bf then}
Consider the cases that $r_i$ is in the middle or the right robot of three neighbor robots. Calculate $x$, the maximum distance that $r_i$ can move vertical to the line $c_i$,$c_l$ and with direction to the inside of the convex hull without causing any two other robots to be on a straight line
 or the distance between any line $c_l$,$c_r$ to $c_m$(any setting of three neighbor robots that involve $r_i$) be less than $\frac{1}{2n}$. $x'$ is the smallest between $\frac{1}{2n}-\epsilon$ and $x$. Consider the line that is vertical to the line $c_i$,$c_l$, starts from $c_i$ and has direction to the inside of the convex hull with an ending point $p$ with distance $x'$ from $c_i$. 
{\bf Return} p.
\end{itemize}
\item {\bf If} $r_i$ belongs to a set of continuous robots that the previous condition is true and the last robot on the right cannot move, {\bf then do} exactly the same step as above but instead of the $r_i$ to be the right robot of three neighboring robots, to be the left with its right neighbors.
\item {\bf If} there exists a configuration of three neighbor robots $c_r$, $c_m$ and $c_l$ that the distance between the line $c_i$,$c_l$ and the point $c_m$ is $<\frac{1}{n}$, {\bf then} {\bf Return} $c_i$.
\item {\bf If} there is a $c_k\in V_i-\{c_i\}$ and a $c_j\in V_i-\{c_i\}$
that their unit disc is tangent with your unit disc, at your left and right respectively,
{\bf then} {\bf return} $c_i$ ($r_i$'s current position).\vspace{-.7em}
\item{\bf If} All robots form one component and $r_i$ does not touch any other robot {\bf then} Consider $c_l$ the center of the left neighbor of $r_i$ and $c_r$ the center of the right neighbor of $r_i$. Draw a line starting from $c_i$, vertical to line $c_l$,$c_r$ with direction to the inside of the convex hull and ending point $p$, $p$ is in distance $\frac{1}{2n}$ from $c_i$. {\bf Return} $p$.
\item {\bf Else} {\bf call} function {\tt In-Largest-Component} with inputs, $V_i$ 
and $c_i$. \vspace{-.7em}
\item {\bf If} function {\tt In-Largest-Component} returns $1$, {\bf then}
{\bf return} $c_i$.\vspace{-.7em}
\item {\bf ElseIf} function {\tt In-Largest-Component} returns $2$, {\bf then}
{\bf call} function {\tt How-Much-Distance} with inputs, $V_i$ 
and $c_i$.\vspace{-.7em}
\begin{itemize}
\item {\bf If} function {\tt How-Much-Distance} returns $1$, {\bf then}
{\bf call} function {\tt Move-to-Point} with inputs, $c_i$ and $c_j$, where
$c_j$ is $c_i$'s right neighbor on the Convex Hull. {\bf Return} the point returned
by {\tt Move-to-Point}.\vspace{-.3em}
\item {\bf ElseIf} function {\tt How-Much-Distance} returns $2$, {\bf then}
{\bf call} function {\tt Connected-Components} with inputs, $V_i$ and $c_i$.
Let $(c_l,c_r)$ be the component that $c_i$ belongs ($c_l$ is the center
of the left-most robot and $c_r$ the center of the right-most robot of the component). 
Draw a straight line between $c_l$ and $c_r$; call it $\overline{AB}$. Draw a parallel
line (wrt $\overline{AB}$) such that it goes through $c_i$; call $c_i$, $C$ and the line 
$\overline{CD}$. Draw a vertical line (wrt $\overline{CD}$) from point $C$ towards the 
inside of Convex Hull with distance $\frac{1}{2n}-\epsilon$. {\bf If} moving to $D$ does not cause 
robot $r_i$ to touch another unit disc in $r_i$'s component, {\tt or}, {\bf if} $c_i$ is $c_l$ or $c_r$, {\bf then} {\bf Return} $D$,
{\bf else} {\bf return} $c_i$ (current position).\vspace{-.3em}
\item {\bf Else} {\bf Return} $c_i$.\vspace{-.7em}
\end{itemize}
\item {\bf ElseIf} function {\tt In-Largest-Component} returns $3$, {\bf then}
{\bf call} function {\tt In-Smallest-Component} with inputs, $V_i$ 
and $c_i$.\vspace{-.7em}
\begin{itemize}
\item {\bf If} function {\tt In-Smallest-Component} returns $1$, {\bf then}
{\bf call} function {\tt Move-to-Point} with inputs, $c_i$ and $c_j$, where
$c_j$ is $c_i$'s right neighbor on the Convex Hull. {\bf Return} the point returned
by {\tt Move-to-Point}.\vspace{-.3em}
\item {\bf ElseIf} function {\tt In-Smallest-Component} returns $2$, {\bf then}
{\bf call} function {\tt How-Much-Distance} with inputs, $V_i$ 
and $c_i$.\vspace{-.3em}
\begin{itemize}
\item {\bf If} function {\tt How-Much-Distance} returns $1$, {\bf then}
{\bf call} function {\tt Move-to-Point} with inputs, $c_i$ and $c_j$, where
$c_j$ is $c_i$'s right neighbor on the Convex Hull. {\bf Return} the point returned
by {\tt Move-to-Point}.\vspace{-.3em}
\item {\bf ElseIf} function {\tt How-Much-Distance} returns $2$, {\bf then}
{\bf call} function {\tt Connected-Components} with inputs, $V_i$ and $c_i$.
Let $(c_l,c_r)$ be the component that $c_i$ belongs ($c_l$ is the center
of the left-most robot and $c_r$ the center of the right-most robot of the component). 
Draw a straight line between $c_l$ and $c_r$; call it $\overline{AB}$. Draw a parallel
line (wrt $\overline{AB}$) such that it goes through $c_i$; call $c_i$, $C$ and the line 
$\overline{CD}$. Draw a vertical line (wrt $\overline{CD}$) from point $C$ towards the 
inside of Convex Hull with distance $\frac{1}{2n}-\epsilon$. {\bf Return} $D$.
\vspace{-.3em}
\item {\bf Else} {\bf Return} $c_i$.\vspace{-.3em}
\end{itemize}
\item {\bf Else} {\bf Return} $c_i$.
\end{itemize}
\end{itemize}
\hrule

\begin{lemma}
\label{NotConnectedLemma}
The point returned by {\tt NotConnected}{\bf (Compute.$\langle$NotConnected$\rangle$)} keeps $V_i$ as a fully visible configuration and $|\CH(V_i)|=n$. 
%not result {\tt G} to be not $fully$ $visible$ and will eventually result all robots to form a $connected$ configuration. 
\end{lemma}

\begin{proof}
Based on Lemma ~\ref{OnConvexHullLemma}, $|\CH(V_i)|=n$ and all robots have full visibility (that is, $V_i$ is a fully visible configuration). 
Based on Lemma ~\ref{AllOnConvexHullLemma}, $V_i$ is not a $connected$ configuration. We now show that in each of the following 6 possible cases,
the point returned keeps $V_i$ as a fully visible configuration:
\begin{enumerate}
\item Robot $r_i$ touches one robot on its left and one robot on its right on the convex hull.

In this case, procedure {\tt NotConnected} will return $r_i$'s current position and hence it will not cause any change.

\item Robot $r_i$ is in the component with the most robots.

In this case, procedure {\tt NotConnected} will return $r_i$'s current position hence it will not cause any change.

\item All the components have the same number of robots and $r_i$ is the rightmost robot of the component that has the smallest distance with its neighbor component, among any distance between any two adjacent components.

In this case, procedure {\tt NotConnected} will call the function {\tt Move-To-Point}, which will return a position on convex hull adjacent to the leftmost robot of the right neighbor component. Hence in this case robot $r_i$'s next position will be on the convex hull, $r_i$ will have full visibility, will not be in the same line with any two other robots and will not cause any other robot not to be on the convex hull, provided that the other robots do not move.

\item All the components have the same number of robots and all the spaces between adjacent components are of the same distance.

In this case, procedure {\tt NotConnected} will cause the components to approach by making small steps. 
By close investigation of the code of this Function, it follows that the point returned in this case, will not cause  $V_i$ to be not $fully$ $visible$ and no three robots will be on the same line, provided that the other robots do not move.

\item All the components have the same number of robots and $r_i$ is not the rightmost robot of the component that has the smallest distance with its neighbor component, among any distance between any two adjacent components, and not all the spaces between adjacent components are of the same distance.

In this case, Procedure {\tt NotConnected} will return $r_i$'s current position hence it will not cause any change.

\item Robot $r_i$ is part of the component that has the smallest number of robots.

In this case, procedure {\tt NotConnected} will call the function {\tt Move-To-Point}, which will return a position on convex hull adjacent to the leftmost robot of the right neighbor component. Hence in this case robot $r_i$'s next position will be on the convex hull, $r_i$ will have full visibility, will not be in the same line with any two other robots and will not cause any other robot not to be on the convex hull, provided that the other robots do not move.\hfill$\blacksquare$ 
%\vspace{-1.8em}
\end{enumerate}
\end{proof}

\subsubsection{Procedure {\tt NotAllOnConvexHull}}

\hrule\vspace{3pt}
\noindent Procedure: {\tt NotAllOnConvexHull}\vspace{3pt}
\hrule\vspace{5pt}
\noindent {\em Precondition}: state= {\bf Compute.NotAllOnConvexHull}\\
\noindent {\em Effect}:\vspace{-.8em}
\begin{itemize}
\item Consider three points $c_l, c_m, c_r\in \CH(V_i)$, where $c_l$ is the left neighbor
and $c_r$ the right neighbor of $c_m$, respectively, on the convex hull, and $c_i$ is one of these 
points.\vspace{-.7em}
\item Draw line segment $\overline{c_lc_r}$. Let $\overline{AB}$ be the line that is vertical 
to $\overline{c_lc_r}$, it goes through $c_l$ and both $\overline{Ac_l}$ and 
$\overline{c_lB}$ have length $\frac{1}{n}$. Line $\overline{CD}$ is defined similarly for $c_r$. 
(See Figure~\ref{fig:morespace} for an example.)\vspace{-.7em}
\item Consider all three cases, that is, $c_i = c_l$, $c_i=c_m$ or $c_i=c_r$. \vspace{-.7em}
\item {\bf If} in {\em any} of the three cases, $c_m$ is in the rectangle $ABCD$, {\bf then} state:= {\bf Compute.OnStraightLine}\vspace{-.7em}
\item [] {\bf Else} state:= {\bf Compute.NotOnStraightLine}\vspace{-.4em} 
\end{itemize}
\hrule

\begin{lemma}
\label{NotAllOnConvexHullLemma}
\sloppy{
{\tt NotAllOnConvexHull}{\bf(Compute.$\langle$NotAllOnConvexHull$\rangle$)} = {\bf Compute.$\langle$OnStraightLine$\rangle$} iff $r_i$ is on the same line with any two other robots that are also on the convex hull.} 
\end{lemma}

\begin{proof}
Based on Lemma ~\ref{StartLemma}, robot $r_i \in \CH(V_i)$. Based on Lemma ~\ref{OnConvexHullLemma}, $|V_i|\not=n$, or $|\CH(V_i)|<n$ or not all robots have full visibility.
Procedure {\tt NotAllOnConvexHull} uses simple geometric calculations to determine if $r_i$ is on straight line with its neighbor robots on $\CH(V_i)$.
\hfill$\blacksquare$
\end{proof}

\subsubsection{Procedure {\tt NotOnStraightLine}}
\hrule\vspace{3pt}
\noindent Procedure: {\tt NotOnStraightLine}\vspace{3pt}
\hrule\vspace{5pt}
\noindent {\em Precondition}: state= {\bf Compute.NotOnStraightLine}\\
\noindent {\em Effect}:\vspace{-.8em}
\begin{itemize}
\item {\bf If} $|\CH(V_i)|=n$ {\bf then} state:= {\bf Compute.SpaceForMore} and {\bf Return}\vspace{-.7em}
\item{\bf If} $|V_i|=n$ {\bf then}
\begin{itemize}
\item {\bf Check} whether there exists a side on the Convex Hull with length of at least 2.  \vspace{-.7em}
\item {\bf If} there exists {\bf then} state:= {\bf Compute.SpaceForMore}\vspace{-.7em}
\item[] {\bf Else} state:= {\bf Compute.NoSpaceForMore}\vspace{-.4em}
\end{itemize}
\item {\bf Else}
\begin{itemize}
\item $\forall c_j \in \CH(V_i)$ {\bf Copy} $c_j$ in a new set named $onCH2$
\item $\forall c_j \notin \CH(V_i)$ draw a straight line from $r_i$ that has an ending point $x$, $x\in\CH$ and $c_j$ is on that line.
\item {\bf Copy} $x$ on $onCH2$
\item {\bf Check} whether there exists a side on $onCH2$ with length of at least 2. 
\item {\bf If} there exists {\bf then} state:= {\bf Compute.SpaceForMore}
\item[] {\bf Else} state:= {\bf Compute.NoSpaceForMore}
\end{itemize}
\end{itemize}
\hrule

\begin{lemma}
\label{NotOnStraightLineLemma}
{\tt NotOnStraightLine}{\bf (Compute.$\langle$NotOnStraightLine$\rangle$)} = {\bf Compute.$\langle$SpaceForMore$\rangle$} {\em iff} $|\CH(V_i)|=n$ or there exist a space, for at least one robot, between any two adjacent robots that are on the convex hull. 
\end{lemma}

\begin{proof}
Based on Lemma ~\ref{StartLemma}, $r_i \in \CH$. Based on Lemma ~\ref{OnConvexHullLemma}, $|V_i|\not=n$, or $|\CH(V_i)|<n$ or not all robots have full visibility. Based on Lemma ~\ref{NotAllOnConvexHullLemma}, $r_i$ is not on a straight line with any two other robots that $\in \CH$.
There are three possible cases:
\begin{enumerate}
\item $|\CH(V_i)|=n$.
In this case, robot $r_i$ moves to the state {\bf Compute.SpaceForMore}, because it is not necessary to create extra space on the convex hull for more robots, because all $n$ robots are already on convex hull.
\item $|\CH(V_i)|<n$ and there exist a space, for at least one robot, on the convex hull.

In this case, procedure {\tt NotOnStraightLine} uses simple calculations to determine if there exist a space for at least one robot on the convex hull. It calculates the distance between adjacent robots on the convex hull and if there exist at least two adjacent points that have more than 2 distance, the correct state to move is {\bf Compute.SpaceForMore}.
\item $|\CH(V_i)|<n$ and no space, for at least one robot, on the convex hull exists.

In this case procedure {\tt NotOnStraightLinel} uses simple calculations to determine if there exist a space for at least one robot on the convex hull. It calculates the distance between adjacent robots on the convex hull and if no such space exists, the correct state to move 
is {\bf Compute.SpaceForMore}.\hfill$\blacksquare$
\end{enumerate}
\end{proof}

\subsubsection{Procedure {\tt SpaceForMore}}
\hrule\vspace{3pt}
\noindent Procedure: {\tt SpaceForMore}\vspace{3pt}
\hrule\vspace{5pt}
\noindent {\em Precondition}: state= {\bf Compute.SpaceForMore}\\
\noindent {\em Effect}:\vspace{-.8em}
\begin{itemize}
\item{\bf If} $r_i$ is tangent with a robot $r_j$, $r_j \in\CH(V_i)$  that they are not adjacent on $\CH(V_i)$ {\bf then} consider $c_l$ and $c_r$, the centers of $r_i$s left and right neighbor on the convex hull respectively. Draw a straight line starting from $c_i$, vertical to the line $c_l$,$c_r$, with direction outside of the convex hull and at distance $\frac{1}{2n}-\epsilon$. The ending point of this line is $p$. {\bf Return} $p$.
\item {\bf Else Return} $c_i$
\end{itemize}

\hrule

\paragraph {} The reason that $p$ is outside of the convex hull by a distance $\frac{1}{2n}-\epsilon$ is because if two robots are not adjacent on the convex hull and are touching, it means that it is possible to obstruct other robots from seeing each-other.

\begin{lemma}
\label{SpaceForMoreLemma}
{\tt SpaceForMore}{\bf (Compute.$\langle$SpaceForMore$\rangle$)} = $c_i$ iff $r_i$ is not  tangent with any robot $r_j$, $r_j \in\CH(V_i)$  that they are not adjacent on $\CH(V_i)$, {\bf else} {\tt SpaceForMore}{\bf (Compute.$\langle$SpaceForMore$\rangle$)} = $p$  were $p$ is at distance $\frac{1}{2n}-\epsilon$ outside of $\CH(V_i)$.
\end{lemma}

\begin{proof}
Based on Lemma ~\ref{StartLemma}, $r_i \in \CH(V_i)$. Based on Lemma ~\ref{OnConvexHullLemma}, $|V_i|\not=n$, or $|\CH(V_i)|<n$ or not all robots have full visibility. Based on Lemma ~\ref{NotAllOnConvexHullLemma}, $r_i$ is not on a straight line with any two other robots that $\in \CH(V_i)$. Based on Lemma ~\ref{NotOnStraightLineLemma}, there exist a space for at least one robot on $\CH(V_i)$. There are are two cases:

 If $r_i$ is not  tangent with any robot $r_j$ such that $r_j \in\CH(V_i)$  and $r_j$ is not adjacent to $r_i$ on $\CH(V_i)$, then  procedure {\tt SpaceForMore} returns $c_i$, else it returns a point $p$ at distance $\frac{1}{2n}-\epsilon$ outside of $\CH(V_i)$.

\end{proof}

\hfill$\blacksquare$

\subsubsection{Procedure {\tt NoSpaceForMore}}
\hrule\vspace{3pt}
\noindent Procedure: {\tt NoSpaceForMore}\vspace{3pt}
\hrule\vspace{5pt}
\noindent {\em Precondition}: state= {\bf Compute.NoSpaceForMore}\\
\noindent {\em Effect}:\vspace{-.8em}
\begin{itemize}
\item Let $c_l$ be the center of $c_i$'s left neighbor 
and let $c_r$ be the center of $c_i$'s right neighbor on 
the Convex Hull. Draw a straight line between $c_l$ and
$c_r$; call it $\overline{AB}$.\vspace{-.7em} 
\item Let $m$ be the center of line $\overline{AB}$. \vspace{-.7em}
\item Draw a vertical line (wrt $\overline{AB}$) starting from $m$ and
ending at distance $\frac{1}{2n} - \epsilon$ away from the
Convex Hull; call the ending point, $p$.\vspace{-.7em}
\item Calculate the point on the line between $m$ and $p$, which $r_i$ can move to maximum distance from $m$ without causing $\CH(V_i)$ to change; call this point $p'$.
\item {\bf Return} $p'$.\vspace{-.4em} 
\end{itemize}
\hrule

\begin{lemma}
\label{NoSpaceForMoreLemma}
{\tt NoSpaceForMore}{\bf (Compute.$\langle$NoSpaceForMore$\rangle$)} = $p$ were $p$ is at distance $\frac{1}{2n}- \epsilon$ outside of $\CH(V_i)$.
\end{lemma}

\begin{proof}
Based on Lemma ~\ref{StartLemma}, $r_i \in \CH$. Based on Lemma ~\ref{OnConvexHullLemma}, $|V_i|\not=n$, or $|\CH(V_i)|<n$ or not all robots have full visibility. Based on Lemma ~\ref{NotAllOnConvexHullLemma}, $r_i$ is not on a straight line with any two other robots that $\in \CH$. Based on Lemma ~\ref{NotOnStraightLineLemma}, no space exists for at least one robot on $\CH(V_i)$.

The correctness of the Lemma follows from the code of the procedure.
\hfill$\blacksquare$
\end{proof}

\subsubsection{Procedure {\tt OnStraightLine}}
\hrule\vspace{3pt}
\noindent Procedure: {\tt OnStraightLine}\vspace{3pt}
\hrule\vspace{5pt}
\noindent {\em Precondition}: state= {\bf Compute.OnStraightLine}\\
\noindent {\em Effect}:\vspace{-.8em}
\begin{itemize}
\item Consider the same setting as in procedure {\tt NotAllOnConvexHull}.\vspace{-.7em}
\item {\bf If} for one of the cases $c_m$ is in the rectangle $ABCD$ and holds that $c_m=c_i$, {\bf then} state:= {\bf Compute.SeeTwoRobot}\vspace{-.7em}
\item [] {\bf Else} state:= {\bf Compute.SeeOneRobot}\vspace{-.4em} 
\end{itemize}
\hrule
\begin{lemma}
\label{OnStraightLineLemma}
{\tt OnStraightLine}{\bf (Compute.$\langle$OnStraightLine$\rangle$)} = {\bf Compute.$\langle$SeeTwoRobots$\rangle$} iff $r_i$ is on the same line
with two robots on the convex hull, its left neighbor $r_l$, and its right neighbor $r_r$. 
\end{lemma}

\begin{proof}
Based on Lemma ~\ref{StartLemma}, robot $r_i \in \CH(V_i)$. Based on Lemma ~\ref{OnConvexHullLemma}, $|V_i|\not=n$, or $|\CH(V_i)|<n$ or not all robots have full visibility. Based on Lemma ~\ref{NotAllOnConvexHullLemma}, $r_i$ is on straight line with its neighbor robots on the convex hull.
Procedure {\tt OnStraightLine} uses simple geometric calculations %based on figure \ref{fig:morespace} 
to determine if robot $r_i$ is in the middle of $c_l$ and $c_r$.
\hfill$\blacksquare$
\end{proof}

\subsubsection{Procedure {\tt SeeOneRobot}}
\hrule\vspace{3pt}
\noindent Procedure: {\tt SeeOneRobot}\vspace{3pt}
\hrule\vspace{5pt}
\noindent {\em Precondition}: state= {\bf Compute.SeeOneRobot}\\
\noindent {\em Effect}:\vspace{-.8em}
\begin{itemize}
\item {\bf Return} $c_i$ (current position).\vspace{-.4em} 
\end{itemize}
\hrule\vspace{.7em}

Then, trivially:

\begin{lemma}
\label{SeeOneRobotLemma}
{\tt SeeOneRobot}{\bf (Compute.$\langle$SeeOneRobot$\rangle$)} = $c_i$.
\end{lemma}

%\begin{proof}
%Trivial
%\end{proof}
%\hfill$\blacksquare$

\subsubsection{Procedure {\tt SeeTwoRobot}}
\hrule\vspace{3pt}
\noindent Procedure: {\tt SeeTwoRobot}\vspace{3pt}
\hrule\vspace{5pt}
\noindent {\em Precondition}: state= {\bf Compute.SeeTwoRobot}\\
\noindent {\em Effect}:\vspace{-.8em}
\begin{itemize}
\item Consider the same setting as in procedure {\tt NotAllOnCOnvexHull} and $c_m=c_i$.\vspace{-.7em} 
\item Consider the line segment that is vertical to 
line $\overline{c_lc_r}$, it starts from $c_i$ with direction outside of the convex hull (if this is not possible to determine choose a random direction) 
and ends at distance $\frac{1}{2n}-\epsilon$ from $c_i$. Call the ending point, $p$. 
Consider the line segment that is vertical to line $\overline{c_lc_r}$, it starts from line $\overline{c_lc_r}$ and ends at distance $\frac{1}{n}$ from line $\overline{c_lc_r}$. Call the ending point, $p'$, such that $c_i$ is in the same line with and between $p'$ and  line $\overline{c_lc_r}$. \vspace{-.7em} 
\item {\bf Return} the point that is nearest to $c_i$, between $p$ and $p'$.\vspace{-.4em} 
\end{itemize}
\hrule

\begin{lemma}
\label{SeeTwoRobotLemma}
The point $p$ returned by {\tt SeeTwoRobot}{\bf (Compute.$\langle$SeeTwoRobot$\rangle$)} is such that if robot $r_i$ moves there
($c_i$ is on $p$), then $r_i$
will no longer be in a straight line with it's two adjacent robots on the convex hull, provided that the other robots do not move. 
\end{lemma}
\begin{proof}
Based on Lemma ~\ref{StartLemma}, robot $r_i \in \CH(V_i)$. Based on Lemma ~\ref{OnConvexHullLemma}, $|V_i|\not=n$, or $|\CH(V_i)|<n$ or not all robots have full visibility. Based on Lemma ~\ref{NotAllOnConvexHullLemma}, $r_i$ is on straight line with its neighbor robots on the convex hull. Based on Lemma ~\ref{OnStraightLineLemma}, $r_i$ is on the same line and in the middle of its left neighbor robot on the convex hull $r_l$ and with its right neighbor robot on the convex hull $r_r$.
Procedure {\tt SeeTwoRobot} results robot $r_i$ to move at distance $\frac{1}{2n}-\epsilon$ from its current position, with direction out of the convex hull as described in the Procedure {\tt SeeTwoRobot}. After moving in this position robot $r_i$ will no longer be on a straight line with its adjacent robots on the convex hull.~$\blacksquare$
\end{proof}

\begin{figure}[t]
\centering
\includegraphics[width=3in]{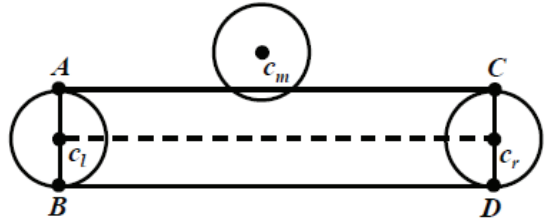}\vspace{-1em}
\caption{An example where the unit disc (robot) with center $c_m$ intersects rectangle $ABCD$.}
\label{fig:morespace}
\end{figure}

\subsubsection{Procedure {\tt NotOnConvexHull}}
\hrule\vspace{3pt}
\noindent Procedure: {\tt NotOnConvexHull}\vspace{3pt}
\hrule\vspace{5pt}
\noindent {\em Precondition}: state= {\bf Compute.NotOnConvexHull}\\
\noindent {\em Effect}:\vspace{-.8em}
\begin{itemize}
%\item Compute Convex Hull: {\bf Set} $CH = \emptyset$. For each $c_j\in V_i-\{c_i\}$ {\bf call} 
%function {\tt On-Convex-Hull} with inputs, the local view
%$V_i$ and $c_j$; If the function returns {\sc yes} then $CH = CH \cup \{c_j\}$.\vspace{-.7em}
\item {\bf Check} whether there is a $c_j\in V_i$ so that
the unit disc with center $c_j$ is tangent with the unit disc with center $c_i$. \vspace{-.7em}
\item {\bf If} yes, {\bf then} state:= {\bf Compute.IsTouching}\vspace{-.7em}
\item [] {\bf Else} state:= {\bf Compute.NotTouching}\vspace{-.4em} 
\end{itemize}
\hrule

\begin{lemma}
\label{NotOnConvexHullLemma}
{\tt NotOnConvexHull}{\bf(Compute.$\langle$NotOnConvexHull$\rangle$)} = {\bf Compute.$\langle$IsTouching$\rangle$}  iff $r_i$'s unit disk is tangent with a unit disk of any other robot. 
\end{lemma}

\begin{proof}
Based on Lemma~\ref{StartLemma}, robot $r_i \notin \CH(V_i)$. Now the claim of the lemma follows by close investigation of procedure {\tt NotOnConvexHull}.
\hfill$\blacksquare$
\end{proof}

\subsubsection{Procedure {\tt IsTouching}}
\hrule\vspace{3pt}
\noindent Procedure: {\tt IsTouching}\vspace{3pt}
\hrule\vspace{5pt}
\noindent {\em Precondition}: state= {\bf Compute.IsTouching}\\
\noindent {\em Effect}:\vspace{-.8em}
\begin{itemize}
\item {\bf Call} function {\tt Find-Points} with input $\CH(V_i)$.\vspace{-.7em}
\item
%{\bf call} function {\tt OnConvexHull} for each $c_j$ that is neighboring $c_i$.
{\bf If} Function {\tt Find-Points} returned one or more points {\bf then} {\bf choose} $p$, the point returned by Function {\tt Find-Points} that is
closest to $c_i$\vspace{-.7em}
\begin{itemize}
\item {\bf If} one of the robots that are touching $r_i$ is closer to $p$
than $r_i$ {\bf then} {\bf return} $c_i$ (current position).\vspace{-.5em}
\item{\bf Else If} one or more of the robots that are touching $r_i$ have the same distance with $r_i$ to $p$  {\bf then}
\begin{itemize}
\item{\bf If} $r_i$ is the rightmost of the robots that are touching and have the same distance to $p$ {\bf then}
$p' \in \CH(V_i)$ and $p'$ is on the line between $c_i$ and $p$.
{\bf Return} $p'$.\vspace{-.1em}
\item{\bf Else}
{\bf return} $c_i$ (current position).
\end{itemize}
\item {\bf Else} $p' \in \CH(V_i)$ and $p'$ is on the line between $c_i$ and $p$. {\bf Return} $p'$ .\vspace{-.7em}   
\end{itemize}

%%%%%%simeia pou tha prokalesoun allagi%%%%%%%%%%%%%%
\item
%{\bf call} function {\tt OnConvexHull} for each $c_j$ that is neighboring $c_i$.
{\bf Else} {\bf choose} the two closest neighboring robots, that have distance of at least 2, to $c_i$ that are on the Convex Hull.\vspace{-.7em}
\begin{itemize}
\item {\bf If} No neighboring robots on the convex hull have distance of at least 2 {\bf then} {\bf return} $c_i$ (current position).
\item {\bf If} one of the robots that are touching $r_i$ is closer to these robots
than $r_i$ {\bf then} {\bf return} $c_i$ (current position).\vspace{-.5em}
\item{\bf Else If} one or more of the robots that are touching $r_i$ have the same distance with $r_i$ to these robots {\bf then}
\begin{itemize}
\item{\bf If} $r_i$ is the rightmost of the robots that are touching and have the same distance to the closest robots on the convex hull {\bf then}
draw a straight line between the centers of the two closest robots to $r_i$ on the convex hull and find the center
of this line, $p$. {\bf Return} $p$.\vspace{-.1em}
\item{\bf Else}
{\bf return} $c_i$ (current position).
\end{itemize}
\item {\bf Else} draw a straight line between the centers of these robots and find the center
of this line, $p$. {\bf Return} $p$.\vspace{-.8em}   
\end{itemize}
\end{itemize}
\hrule\vspace{.5em}

In a given set of robots, we consider that a robot has {\em higher proximity} compared to the other robots of that set if it is the closest to its closest space on the convex hull or to the closest space that Function {\tt FindPoints} returned (depending on the case). If more than one robots of that set of robots have the same distance to the closest space, then the rightmost of these robots has the highest proximity (straight direction is considered to be to the outside of the convex hull of the target point).

\begin{lemma}
\label{IsTouchingLemma}
{\tt IsTouching}{\bf (Compute.$\langle$IsTouching$\rangle$)} will result robot $r_i$'s unit disk to no longer be tangent with any other robot's unit disk (from the robots that $r_i$ touches)  $if$ $r_i$ has the highest proximity (from the robots that are touching). If no space of at least 2 exists on the convex hull, then $r_i$ stays in the same position.
\end{lemma}

\begin{proof}
Based on Lemma ~\ref{StartLemma} $r_i \notin \CH(V_i)$. Based on Lemma ~\ref{NotOnConvexHullLemma} robot $r_i$'s unit disk is tangent with at least one other robot's unit disk.

There are 9 possible cases:
\begin{enumerate}

\item Function {\tt FindPoints} returned one point or more and  $r_i$ is not the closest robot (from the robots that are touching) to $p$.

In this case, $r_i$ will remain in the same position because it does not have the highest proximity.

\item Function {\tt FindPoints} returned one point or more.  $r_i$ is the closest robot (from the robots that are touching) to $p$ and has the same distance with $p$ with at least another robot. $r_i$ is the rightmost robot of the robots that are touching and have the same distance to $p$.

In this case, $r_i$ will move to $p'$ because it has the highest proximity.

\item Function {\tt FindPoints} returned one point or more.  $r_i$ is the closest robot (from the robots that are touching) to $p$ and has the same distance with $p$ with at least another robot. $r_i$ is not the rightmost robot of the robots that are touching and have the same distance to $p$.

In this case, $r_i$ will remain in the same position because it does not have the highest proximity.

\item Function {\tt FindPoints} returned one point or more.  $r_i$ is the closest robot (from the robots that are touching) to $p$ and no other robot (from the robots that are touching)  has the same distance to $p$ with $r_i$.

In this case, $r_i$ will move to $p'$ because it has the highest proximity.
%%%%%%%%%

\item Function {\tt FindPoints} did not returned any point. No space of at least 2 exists on the convex hull.

In this case, $r_i$ will remain in the same position because not enough space for it exists on the convex hull.

\item Function {\tt FindPoints} did not returned any point.  $r_i$ is not the closest robot (from the robots that are touching) to the two closest robots (from $r_i$) on the convex hull.

In this case, $r_i$ will remain in the same position because it does not have the highest proximity.

\item Function {\tt FindPoints} did not returned any point.  $r_i$ is the closest robot (from the robots that are touching) to the two closest robots (from $r_i$) on the convex hull and has the same distance with at least another robot. $r_i$ is the rightmost robot of the robots that are touching and have the same distance to the two closest robots (from $r_i$) on the convex hull.

In this case, $r_i$ will move to the center of the line between those two robots because it has the highest proximity.

\item Function {\tt FindPoints} did not returned any point.  $r_i$ is the closest robot (from the robots that are touching) to the two closest robots (from $r_i$) on the convex hull and has the same distance with $p$ with at least another robot. $r_i$ is not the rightmost robot of the robots that are touching and have the same distance to the two closest robots (from $r_i$) on the convex hull.

In this case, $r_i$ will remain in the same position because it does not have the highest proximity.

\item Function {\tt FindPoints} did not returned any point.  $r_i$ is the closest robot (from the robots that are touching) to the two closest robots on the convex hull and no other robot (from the robots that are touching)  has the same distance to those robots with $r_i$.

In this case, $r_i$ will move to the center of the line between those two robots because it has the highest proximity.\hfill$\blacksquare$
\end{enumerate}
\end{proof}

\begin{lemma}
\label{IsTouchingLemma2}
{\tt IsTouching}{\bf (Compute.$\langle$IsTouching$\rangle$)} will result at least one of the robots that are touching to move with direction to the convex hull, if a space of at least 2 exists on the convex hull.
\end{lemma}

\begin{proof}
The claim follows from the code of the procedure.\hfill$\blacksquare$
\end{proof}

\subsubsection{Procedure {\tt NotTouching}}
\hrule\vspace{3pt}
\noindent Procedure: {\tt NotTouching}\vspace{3pt}
\hrule\vspace{5pt}
\noindent {\em Precondition}: state= {\bf Compute.NotTouching}\\
\noindent {\em Effect}:\vspace{-.8em}
\begin{itemize}
\item {\bf Call} function {\tt Find-Points} with input $\CH(V_i)$. \vspace{-.7em}
\item {\bf If} function {\tt Find-Points} returns at least one point,
%in which you can move at without causing the Convex Hull to change, 
{\bf then} state:= {\bf Compute.NotChange}\vspace{-.7em}
\item [] {\bf Else} state:= {\bf Compute.ToChange}\vspace{-.4em} 
\end{itemize}
\hrule

\begin{lemma}
\label{NotTouchingLemma}
{\tt NotTouching}{\bf (Compute.$\langle$NotTouching$\rangle$)} = {\bf Compute.$\langle$NotChange$\rangle$}  iff $r_i$ can move on the convex hull, without causing any additional change on the convex hull. 
\end{lemma}

\begin{proof}
Based on Lemma ~\ref{StartLemma}, robot $r_i \notin \CH(V_i)$. Based on Lemma ~\ref{NotOnConvexHullLemma}, robot $r_i$ is not touching with any other robot. Procedure {\tt NotTouching}
calls Function {\tt Find-Points}. Per Lemma ~\ref{FindPointsLemma}, function {\tt Find-Points} returns all the possible points on convex hull, that $r_i$ can move to,
without causing any changes to $\CH(V_i)$, provided that other robots do not move. If Function {\tt Find Points} does not return any points, the next state will correctly
be {\bf Compute.ToChange}. Otherwise, the correct state to move is {\bf Compute.NotChange}.
\hfill$\blacksquare$
\end{proof}

\subsubsection{Procedure {\tt ToChange}}
\hrule\vspace{3pt}
\noindent Procedure: {\tt ToChange}\vspace{3pt}
\hrule\vspace{5pt}
\noindent {\em Precondition}: state= {\bf Compute.ToChange}\\
\noindent {\em Effect}:\vspace{-.8em}
\begin{itemize}
\item 
%{\bf Call} function {\tt OnConvexHull} for each $c_j$ that is neighboring $c_i$
%and have distance more than~$2$. 
Among the robots that are neighboring on the convex hull and have distance at least 2 from each-other, {\bf choose} the two closest ones to $r_i$.\vspace{-.7em}
\item {\bf If} no neighbor robots have distance greater or equal with 2, {\bf Return} $c_i$
\item {\bf Else} Draw a straight line between the centers of these robots and find the center
of this line, $p$. {\bf Return} $p$.\vspace{-.4em} 
\end{itemize}
\hrule

\begin{lemma}
\label{ToChangeLemma}
{\tt ToChange}{\bf (Compute.$\langle$ToChange$\rangle$)} = $p$, when $p \in \CH(V_i)$ if there exists a space of at least 2 on the convex hull. Else {\tt ToChange}{\bf (Compute.$\langle$ToChange$\rangle$)} = $c_i$
\end{lemma}

\begin{proof}
Based on Lemma ~\ref{StartLemma}, $r_i \notin \CH(V_i)$. Based on Lemma ~\ref{NotOnConvexHullLemma}, $r_i$ is not tangent with any other robot. Based on Lemma ~\ref{NotTouchingLemma}, $r_i$ cannot move to the convex hull without causing it to change. The claim now follows from the code of the procedure. \hfill$\blacksquare$
\end{proof}

\subsubsection{Procedure {\tt NotChange}}
\hrule\vspace{3pt}
\noindent Procedure: {\tt NotChange}\vspace{3pt}
\hrule\vspace{5pt}
\noindent {\em Precondition}: state= {\bf Compute.NotChange}\\
\noindent {\em Effect}:\vspace{-.8em}
\begin{itemize}
\item {\bf Call} function {\tt Find-Points} with input $\CH(V_i)$.\vspace{-.7em}
\item {\bf Choose} the point returned by function {\tt Find-Points} that is
closest to $c_i$; call it $x$.\vspace{-.7em}
\item $p \in \CH$ and $p$ is in the line between the points $c_i$ and $x$. 
\item {\bf Return} $p$.\vspace{-.4em} 
\end{itemize}
\hrule

\begin{lemma}
\label{NotChangeLemma}
{\tt NotChange}{\bf (Compute.$\langle$NotChange$\rangle$)} = $p$, where $p \in \CH(V_i)$.
\end{lemma}

\begin{proof}
Based on Lemma ~\ref{StartLemma}, $r_i \notin \CH(V_i)$. Based on Lemma ~\ref{NotOnConvexHullLemma}, $r_i$ is not tangent with any other robot. Based on Lemma ~\ref{NotTouchingLemma}, $r_i$ can move to the convex hull without causing it to change.
Then it follows that the returned point is on the convex hull. \hfill$\blacksquare$
\end{proof}

%%%%%%%%%%%%%%%%%%%%%%%%%%%%%%%%%% SECTION DISTRIBUTED %%%%%%%%%%%%%%%
\input{Distributed}

\section{Conclusions} In this paper we have considered the problem of gathering non-transparent, fat robots
in an asynchronous setting. We have formulated the problem and the model with a state-machine representation
and developed a Distributed Algorithm that solves the problem for any number of robots. The correctness of our
algorithm relies on the assumption of chilarity~\cite{robotsbook} (robots agree on the orientation of the axes of
their local coordination system).
This is the only assumption we needed to add to the model considered in~\cite{Pelc09}. We believe this
is a very small price to pay in order to solve the gathering problem for any number of fat robots. Nevertheless, it
would be very interesting to investigate whether one can remove this assumption and still be able to solve the gathering problem 
for any number of fat robots. Certainly one will need to take a different
approach than the one we use in this paper, as our approach depends greatly on this assumption. 
%One of our modeling assumptions is that robots have unlimited range of visibility (as long as their view is
%not blocked by another robot). An interesting future direction would be to investigate whether it is possible 
%to solve the problem under the considered model, but assuming that the robots' range of visibility is limited. 

\end{document}

%% file: intro.tex
\section{Introduction}
\paragraph{Motivation and Prior Work:}
There is an increasing number of applications that could benefit from having a team of autonomous robots
cooperate and complete various tasks in a self-organizing manner. These tasks could, for example, require the robots 
to work in dangerous and harsh environments 
(e.g., for space, underwater or military purposes) or require high accuracy or speed (e.g., in nanotechnology, scientific computing) or be
of research interest (e.g., artificial intelligence). It is usually desirable for the robots to be as simple and cheap as possible 
and have limited computing power, in order to be able to produce them fast in large numbers.
%These robots can be used in space missions, or in military. Furthermore many scientific fields are interested in such robots,
% like: robotics, artificial intelligence, nanotechnology and more recently, distributed algorithms.

A fundamental problem that has drawn much attention in the recent years is {\em Gathering}~\cite{Japan99,Peleg04,Flocchini05,Survey1,Survey2}, where a team of autonomous mobile robots must gather to a certain point or region or form a certain formation (e.g., geometric shapes) in the plane. The problem has been studied under various modeling assumptions. For example, asynchronous, semi-synchronous and synchronous settings have been considered. Robots may have a common coordination system or have common sense of direction and 
use compasses to navigate in the plane, may have stable memory or be history oblivious. A modeling feature that is shared by all prior works considering the asynchronous setting is that robots are equipped with a vision device (e.g., a camera) and operate under the so called Look-Compute-Move cycle. Within a cycle, a robot takes a snapshot of the plane (Look), based on the snapshot it might perform some local computations (Compute), and it might decide to move to some point in the plane (Move). The range of visibility of robots may be limited or unlimited.
We refer the reader to Surveys~\cite{Survey1,Survey2} and the recent monograph~\cite{robotsbook} for a more comprehensive exposition of works on the gathering problem.  

Up until the work of Czyzowicz et al.~\cite{Pelc09}, the gathering problem was considered only under the assumption that  
robots are a point on the plane and are transparent, that is, a robot can see through another
robot. These assumptions do not reflect reality, as real robots are not points, but instead they have
a physical extent. Furthermore, robots are not transparent, that is, robots may block the view of other robots or robots may collide.
Having this in mind, Czyzowicz et al.~\cite{Pelc09} initiated the study of the gathering problem with {\em fat} robots, that
is, non-transparent unit-disks (disks of radius of 1 unit). As fat robots cannot occupy the same space on the plane, the gathering
problem can no longer require robots to gather at the same point. Instead, per~\cite{Pelc09}, {\em gathering fat robots means 
forming a configuration for which the union of all discs representing them is connected}. The model considered in~\cite{Pelc09}
is the following: Robots operate in Look-Compute-Move cycles, they are identical, anonymous, history oblivious, non-transparent, and fat.
They do not share a common coordination system and the only means of coordination is by vision; robots have unlimited visibly, unless
their view is obstructed by another robot. An asynchronous setting is considered, modeled by an adaptive adversary that can stop 
any robot for finite time, control the ``speed'' of any robot or cause robots moving into intersecting trajectories to collide. Under this
model, the authors present solutions for the gathering problem for {\em three} and {\em four} robots. The proposed solutions consider exhaustively all 
possible classes of configurations in which robots may be found; a different gathering strategy corresponds to each possible case.
As this approach cannot be generalized for larger number of robots (the cases grow exponentially as the number of robots increases), 
the authors left open the question of whether it is possible to solve gathering
%under their model 
for any collection of $n\ge 5$ fat robots.%\vspace{-1em} 

\paragraph{Our Contribution:} In this work we provide a positive answer to the above question. In particular, we consider the model 
of~\cite{Pelc09} with the additional assumption of {\em chirality}~\cite{robotsbook} (i.e., robots agree on the orientation of the axes of their local 
coordinator system) and present a distributed algorithm for the gathering problem for {\em any} number $n$ of fat robots.     

The key feature of our solution is to bring the robots in a configuration of {\em full visibility} where all robots
can see all other robots. However, the power of the adversary and the fact that robots are non-transparent makes this task challenging. 
We have overcome this challenge by requiring robots to aim in forming a convex hull in which all robots 
will be on the convex. During the computation, robots that are on the convex do not move and robots that are inside the convex hull try to move on the convex hull. However, if robots that are on the convex hull realize that they obstruct other robots that are also on the convex hull from seeing each other, then they move with direction outside of the convex hull in such a way, that they no longer cause any obstruction of the view of the other robots. Furthermore, if robots on the convex hull realize that there is no ``enough space'' for robots that are inside the convex to be placed on the convex, they move to a direction outside of the convex to make space. All these  are further complicated 
%{\bf Chrysovalandis: limitations are further complicating the problem} 
due to asynchrony, as robots may have very different local views of the system. We show that eventually the convex hull will ``expand'' in such a way that all robots will be on the convex hull and no three robots will be on the same line. This leads to a configuration that all robots have full visibility. This is the first conceptual phase of the algorithm.

In the second conceptual phase of the algorithm, once all robots have full visibility and they are aware of this, robots start to converge in such a manner that full visibility is not lost (again, asynchrony complicates issues). To do so, robots exploit their knowledge of $n$ and of the common unit of distance (since all robots are unit-disks, this gives them ``for free'' a common measure of distance~\cite{Pelc09}). We show that eventually all robots form a connected configuration and terminate, yielding a solution to the gathering problem.  

The key in successfully proving the correctness of the algorithm is the formulation of the model and the problem
using a state-machine representation. This enabled us to employ typical techniques for proving safety and liveness properties and argue on the state transitions of the robots, which against asynchrony it can be a very challenging task.%\vspace{-1em}

\paragraph{Other Works Considering Fat Robots.} After the work in~\cite{Pelc09} some attempts were made
in solving the gathering problem with $n\geq 5$ fat robots. However these works consider different models
than the one considered in~\cite{Pelc09}. In~\cite{Chaudhuri10} it is assumed that the fat robots are transparent. This assumption makes the problem significantly easier, as robots have full visibility at all times. As discussed above, having the robots reach a configuration with full visibility was the main challenge in our work. In~\cite{Friedhelm2011} fat robots are non-transparent and have limited visibility, but a synchronous setting is considered. Furthermore, the gathering point is predefined (given as an input to the robots) and the goal is for the robots to gather in an area as close as possible to this point. Two versions of the problem are studied: in continuous space and time, and in discrete space (essentially $\mathbb{Z}^2$) and time. In the continuous case a randomized solution is proposed. In the discrete case the proposed solutions require additional modeling assumptions such as unique robot ids, or direct communication between robots, or randomization. 
The work in~\cite{Dutta12} also considers fat robots with limited visibility, but in an asynchronous setting. In contrast with the model we consider, robots have a common coordination system, that is they agree both on a common origin and axes (called Consistent Compass in~\cite{robotsbook}. The objective of the robots is to gather to a circle with center C, which is given as an input along with the radius of the circle. The common coordination system and the predefined knowledge of the circle to be formed enables the use of geometric techniques that cannot be used in our model. In~\cite{Bolla12} they consider fat robots with limited visibility and without a common coordination system, but in a synchronous setting. Furthermore 
the correctness of the proposed algorithm is not proven analytically, but rather demonstrated via simulations.

%% file: Distributed.tex
\section{Distributed Algorithm for Gathering}
\label{sec:Alggather}

The high level idea of the algorithm is as follows: The objective is
for the robots to form a convex hull and be able to see each other. 
%Once all robots are on the
%convex hull they can see each other and hence 
Once this is achieved, then the robots start to converge (they get closer), while
maintaining the convex hull formation,  so that they form a connected component.
It follows that when all robots are on the convex hull, they can see each other, 
and are connected, the gathering problem is solved and each robot terminates. 

The {\em distributed algorithm} is essentially composed of the asynchronous execution
of the robots' state transition cycle (including the local algorithm when in
state {\bf Compute}).

We now proceed to show that the distributed algorithm correctly solves the
gathering problem. We first provide some important definitions and then
we proceed with the proof of correctness.  Robots make decisions based on
their local views, but due to asynchrony, each robot's local view might not
reflect the current system configuration. Hence, our proof shows that the local
decisions made by the robots are designed in such a way, that robots can
coordinate correctly in the face of asynchrony and hence reach a solution
to the gathering problem. 

\subsection{Definitions}
\label{baddefs}
%We begin with some definitions. 
Given a Robot Configuration ${\cal R}$, we denote by ${\cal G}_{\cal R}$ the geometric configuration 
and by ${\cal S}_{\cal R}$ the state configuration of ${\cal R}$.
Recall that for a geometric configuration ${\cal G}$, we denote by $\con({\cal G})$ the convex hull formed by the points in ${\cal G}$, as output by Graham's Algorithm. Also, we denote by $\CH({\cal G})\subseteq {\cal G}$ the set of points in ${\cal G}$ that are {\em on} the convex hull. 
%
%\noindent Given any execution of the algorithm we say that $V_i \equiv {\cal G}_{{\cal R}_{m}}$ if robot $r_i$ took a snapshot when this configuration %was at place.

\paragraph{Bad Configurations.}
We say that a robot configuration ${\cal R}_x$ is a {\em bad} configuration, when one of the two following cases is true:

\begin{enumerate}
\item Bad configuration of Type 1. When all of the following hold:
\begin{itemize}
\item Configuration ${\cal G}_{{\cal R}_x}$ is {\em fully visible} and $|\CH({\cal G}_{{\cal R}_x})|=n$;
\item A robot $r_i$ in this configuration has as local view $V_i$ a previous configuration ${\cal G}_{{\cal R}_y}$, $y<x$, 
such that $|\CH({\cal G}_{{\cal R}_y})|<n$, $r_i\in\CH({{\cal R}_y})$ and $r_i$ sees that no space for more robots to get on the convex hull exists. % $\CH({\cal G}_{{\cal R}_y})$ exists. 
\end{itemize}	
\item Bad configuration of Type 2. When all of the following hold:
\begin{itemize}
\item Configuration ${\cal G}_{{\cal R}_x}$ is {\em fully visible} and $|\CH({\cal G}_{{\cal R}_x})|=n$;
\item There exists a preceding configuration ${\cal G}_{{\cal R}_y}$, $y<x$, in which at least four robots, call them $r_l, r_{m1}, r_{m2}$ and $r_r$,  are on a straight line and $r_l, r_{m1}, r_{m2},r_r \in \CH({\cal G}_{{\cal R}_y})$.
\end{itemize}
\end{enumerate}	

Both types are considered bad, because they can potentially lead to a succeeding configuration (wrt ${\cal R}_x$) that is no longer
fully visible or all robots are on the convex hull; a property that we would like, once it holds, to continue holding for
all succeeding configurations.

Let us explain how this is possible, first for the bad configuration of type 1. According to the local algorithm,
when robot $r_i$ witness a view as described in the second bullet of type 1 configuration, robot $r_i$ must start moving 
with direction outside of the convex hull so to make space for more robots to get on the convex hull. This is also the case for all
robots sharing the same or similar view with $r_i$. When $r_i$ starts moving (it gets in state {\bf move}),
the adversary can impose the following strategy: It makes $r_i$ to ``move too slow" and lets the other robots move 
with such ``a speed'' that the robots reach configuration ${\cal R}_x$. Since $r_i$ has not changed its state (it is still in state {\bf move}),
it continues to move outside of the convex hull. This may cause a neighboring robot of $r_i$ not to be on the convex hull anymore or not
be able to see all robots. Hence, while ${\cal G}_{{\cal R}_x}$ was a {\em fully$ $visible} configuration and $|\CH({\cal G}_{{\cal R}_x})|=n$,
it is possible for a succeeding configuration not to have one (or both) of the these properties anymore.

Now we consider a type 2 bad configuration. According to the local algorithm, if robots $r_l, r_{m1}, r_{m2},r_r$ witness
configuration ${\cal G}_{{\cal R}_y}$, then robots $r_{m1}$ and $r_{m2}$ must start moving with direction outside of the convex hull
(the robots that realize they are in the middle of the straight line must move outside so to enable the ``edge'' robots to
see each other; the ``edge'' robots do not move). When $r_{m1}$ and $r_{m2}$ start moving (they get in state {\bf move})
the adversary can impose the following strategy: It lets robot $r_{m1}$ to move slightly and then it stops it (with a $stop(r_{m1})$ event).
It lets robot $r_{m2}$ to move slightly and then the adversary makes it to move very slow (so robot $r_{m2}$ is still in state
{\bf move}). The adversary could stop robot $r_{m1}$ and delay $r_{m2}$ in such a way that configuration ${\cal R}_x$ is reached
(recall that $|\CH({\cal G}_{{\cal R}_x})|=n$ and ${\cal G}_{{\cal R}_x}$ is a {\em fully visible} configuration). 
But since $r_{m2}$ continues to move, it is possible to cause robot $r_{m1}$ to no longer be $\in \CH$ or some other
robot (including $r_{m2}$) not be able to see all other robots. Hence it is possible for a succeeding configuration of 
${\cal G}_{{\cal R}_x}$ not to have one (or both) of the these properties anymore.

\paragraph{Safe Configurations.}
We say that a robot configuration ${\cal R}$ is a {\em safe} configuration, when the following is true:
\begin{itemize}
%\item Safe configuration of Type 1. $|\CH({\cal G}_{{\cal R}})|=n$ and the distance between any two neighboring robots on the convex hull is 
%{\em safe} (as defined in Lemma~\ref{CHlengthLemma}).
\item [] $|\CH({\cal G}_{{\cal R}})|=n$, ${\cal G}_{{\cal R}}$ is a {\em fully visible} configuration and 
$\forall r_i$, $|\CH(V_i)|=n$ and $V_i$ is a {\em fully visible} configuration (that is, all robots know that the configuration
is fully visible).
\end{itemize}

The reason we consider these configurations as safe, is because, as we will show later, once an execution of the
algorithm reaches such a safe configuration, then no succeeding configuration can be a bad configuration. 

We define a {\em bad execution fragment} (resp. execution) of the algorithm to be an execution fragment (resp. execution) 
that contains at least one bad robot configuration. Similarly, we define a {\em good execution fragment} (resp. execution)
to be an execution fragment (resp. execution) that contains only good configurations. 

\subsection{Proof of Correctness}
\label{subsec:correctness}

The proof is broken into two parts. In the first part we prove safety and liveness
properties considering only good execution fragments. 
Then we show that the algorithm is correct for any execution (including ones containing bad configurations).

\subsubsection{Good Executions}
\label{goodexecs}

In the section (with the exception of the first lemma) we consider only executions and executions fragments that are good,  
that is, they do not consider bad configurations. We first prove safety and then liveness properties
for such executions.  

\subsubsection*{Safety Properties}

The following lemma states that as long as not all robots are on the convex hull, 
or even if all robots are on the convex hull but there is at least one robot that cannot
see all other robots, then the convex hull can only expand. (This property holds
even for bad execution fragments).

\begin{lemma}
\label{IncreaseLemma}
Given an execution fragment ${\cal R}_{0},e_1, ... ,{\cal R}_{m-1}$ such that for all ${\cal R}_{k}$, $0 \le k \le m-1$ holds that:
c1: $|\CH({\cal G}_{{\cal R}_{k}})|<n$ or c2: $|\CH({\cal G}_{{\cal R}_{k}})|=n$ and ${\cal G}_{{\cal R}_{k}}$ is not a $fully$ $visible$ configuration, then for any step $\langle {\cal R}_{m-1},e_m,  {\cal R}_{m}\rangle$, $\con({\cal G}_{{\cal R}_{m-1}})\subseteq \con({\cal G}_{{\cal R}_{m}})$

\end{lemma}

\begin{proof}
The possible events $e_m$ are:
\begin{enumerate}
\item [(A)] $e_m$ involves (directly) robot $r_i$. If $r_i$ in ${\cal R}_{m-1}$ is in state {\bf Wait}, {\bf Look} or {\bf Compute}, then it trivially holds that none of the possible events $e_m$ can affect the $\con$. 
So, it remains to consider the case that $r_i$ in ${\cal R}_{m-1}$ is in state {\bf Move}.
In this case, there are three possible cases for event $e_m$: $stop(r_i),~arrive(r_i)$ or $collide(X), r_i\in X$.

Since $r_i$ is in state {\bf Move}, then it is following a trajectory ($start$,$target$), where
$start$ is the position of its center when it start moving, and $target$ is the position it wants to reach, as it was
calculated when the robot was in state {\bf Compute} (it is possible that $start = target$), say in ${\cal R}_k,~k<m-1$. 
Furthermore, $r_i$ made decisions based on the view the robot obtained while in state 
{\bf Look}, in some configuration $R_{k'},~k'<k$. It follows that $k'<m-1$, hence the lemma Hypothesis applies (i.e., for ${\cal G}_{{\cal R}_{k'}}$ either property c1 or c2 hold). In other words, for configurations $R_{k'}$ through $R_{m-1}$, $V_i \subseteq {\cal G}_{{\cal R}_{k'}}$. 
Now, for $V_i$ we have the following possible cases:
%Specifically, $r_i$ had a view where c1 or c2 was true for ${\cal R}_k$. 
%Now, for For ${\cal R}_{k'}$, $k'<k<m-1$, we get the following possible cases:
\begin{itemize}
\item $r_i \in \CH(V_i)$. %\CH({\cal G}_{{\cal R}_{k'}})$. 
Based on Lemma~\ref{StartLemma} and Function {\tt Start}, $r_i$ first gets into state {\bf Compute.OnConvexHull}. Then, based on Lemma~\ref{OnConvexHullLemma} and Function {\tt OnConvexHull}, $r_i$ gets into state {\bf Compute.NotAllOnConvexHull}, since c1 or c2 is true for 
${\cal G}_{{\cal R}_{k'}}$. Now the following are possible: %$|\CH({{\cal G}_{{\cal R}_{k}}})|<n$.
\begin{itemize}
\item $r_i$ is on straight line with two other robots that are also on the convex hull. %are on$\in \CH({{\cal G}_{{\cal R}_{k}}})$
In this case, per Lemma~\ref{NotAllOnConvexHullLemma} and Function {\tt NotAllOnConvexHull}, robot $r_i$ gets into state 
{\bf Compute.OnStraightLine}. Then two cases are possible:

\begin{itemize}
\item $r_i$ is in the middle of the two other robots. %that $\in \CH({{\cal G}_{{\cal R}_{k}}})$
Based on Lemma~\ref{OnStraightLineLemma} and Function {\tt OnStraightLine}, robot $r_i$ gets to state {\bf Compute.SeeTwoRobots}
and runs Procedure {\tt SeeTwoRobots}. Based on  Procedure {\tt SeeTwoRobots} and per Lemma ~\ref{SeeTwoRobotLemma}, the procedure returns a point $p$ with direction {\em away} from the convex hull (as witnessed
in view $V_i$ in configuration $R_{k'}$). If $e_m$ is $Stop(r_i)$ or $Collide(X),~r_i\in X$, $r_i$'s position in 
$R_m$ is a point between the trajectory ($c_i,~p$), $c_i$ being the position of $r_i$'s in $V_i$. Hence   
%$r_i$ moves a distance of at least $\delta$ with direction from $c_i$ to $p$,  hence 
$\con({\cal G}_{{\cal R}_{m-1}})$ can only increase (it certainly cannot decrease since it is moving out of the convex hull).
If $e_m = Arrive(r_i)$, then $r_i$ reaches point $p$, which again means that 
$\con({\cal G}_{{\cal R}_{m-1}})$ can only increase.
%\end{itemize}
\item  $r_i$ is not in the middle of the two other robots. % that $\in \CH({{\cal G}_{{\cal R}_{k}}})$
Based on Lemma~\ref{OnStraightLineLemma} and Function {\tt OnStraightLine}, robot $r_i$ gets into state {\bf Compute.SeeOneRobot}
and runs Procedure {\tt SeeOneRobot}. Based on Lemma ~\ref{SeeOneRobotLemma}, the procedure returns $c_i$, that is, the robot does not move. 
Hence $r_i$ does not cause $\con({\cal G}_{{\cal R}_{m-1}})$ to change. 
\end{itemize}
\item $r_i$ is not on a straight line with any two other robots that are also on the convex hull. % $\in \CH({{\cal G}_{{\cal R}_{k}}})$
Per Lemma~\ref{NotAllOnConvexHullLemma} and Function {\tt NotAllOnConvexHull}, robot $r_i$ gets into state {\bf Compute.NotOnStraightLine}.
Then the following cases are possible:
\begin{itemize}
\item Condition c1 holds and $r_i$ sees that there exist enough space for at least one robot to get on the convex hull. 
%on $\CH({{\cal G}_{{\cal R}_{k}}})$
Then, per Lemma~\ref{NotOnStraightLineLemma} and Function {\tt NotOnStraightLine}, robot $r_i$ gets into state {\bf Compute.SpaceForMore}
and runs Procedure {\tt SpaceForMore}, per Lemma ~\ref{SpaceForMoreLemma}, it returns $c_i$, that is, the robot does not move, or it moves with direction outside of the convex hull. Hence $r_i$ does not 
cause $\con({\cal G}_{{\cal R}_{m-1}})$ to change or it causes $\con({\cal G}_{{\cal R}_{m-1}})$ to increase.
\item Condition c1 holds and $r_i$ sees that there is not enough space for at least one robot to get on the convex hull. 
%on $\CH({{\cal G}_{{\cal R}_{k}}})$
In this case, per Lemma~\ref{NotOnStraightLineLemma} and Function {\tt NotOnStraightLine}, robot $r_i$ gets into state 
{\bf Compute.NoSpaceForMore} and runs Procedure {\tt NoSpaceForMore}. Based on Lemma ~\ref{NoSpaceForMoreLemma}, the procedure returns a point $p$ with direction {\em away} 
from the convex hull (as witnessed in view $V_i$ in configuration $R_{k'}$). Then, using the exact reasoning as above (when the Procedure
{\tt SeeTwoRobots} is run), it follows that $\con({\cal G}_{{\cal R}_{m-1}})$ can only increase.
\item The case that Condition c2 holds is handled identically as above, depending what $r_i$ sees. 
\end{itemize}
\end{itemize}
%\end{itemize}
\item $r_i \notin \CH(V_i)$. (Only when condition c1 holds.)
Based on Lemma~\ref{StartLemma} and Function {\tt Start}, $r_i$ gets into state {\bf Compute.NotOnConvexHull}.
Then we have the following cases. 
\begin{itemize}
\item $r_i$ is touching another robot. 
Based on Lemma~\ref{NotOnConvexHullLemma} and Function {\tt NotOnConvexHull}, $r_i$ gets into state {\bf Compute.IsTouching}
and runs Procedure {\tt IsTouching}, based on Lemma ~\ref{IsTouchingLemma}, it returns a point $p\in\CH({\cal G}_{{\cal R}_{k'}})$ (that is, a point towards
the witnessed convex hull) or $c_i$. So, this means, regardless if $e_m$ is a {\em Stop, Collide or Arrive} event on $r_i$, 
robot $r_i$
can reach up to the boundary of $CH({\cal G}_{{\cal R}_{k'}})$. Then it is not difficult to see that $r_i$
does not cause $\con({\cal G}_{{\cal R}_{m-1}})$ to change ($r_i$ will either be on the boundary or inside of $\con({\cal G}_{{\cal R}_{m-1}})$).
%or PThree possible events can happen:
%\begin{itemize}
%\item Stop($r_i$) or Collide($r_i$)
%$r_i$ moves a distance of at least $\delta$ with direction from $c_i$ to $p$, hence  it does not cause $\con({\cal G}_{{\cal R}_{m-1}})$ to change.
%
%\item Arrive($r_i$)
%$r_i$ moves to $p$, hence  it does not cause $\con({\cal G}_{{\cal R}_{m-1}})$ to change.
%\end{itemize}
\item $r_i$ is not touching any other robot.
Based on Lemma~\ref{NotOnConvexHullLemma} and Function {\tt NotOnConvexHull}, $r_i$ moves to state {\bf Compute.NotTouching}. 
\begin{itemize}
\item $r_i$ can move towards $\CH({{\cal G}_{{\cal R}_{k'}}})$ without causing it to change.
Then, per Lemma~\ref{NotTouchingLemma} and Function {\tt NotTouching}, $r_i$ gets into state {\bf Compute.NotChange}
and runs Procedure {\tt NotChange}, based on Lemma ~\ref{NotChangeLemma}, it returns a point $p \in \CH({\cal G}_{{\cal R}_{k'}})$. If $V_i\subseteq {{\cal G}_{{\cal R}_{k}}}$,
then as above, it follows that $r_i$ does not cause $\con({\cal G}_{{\cal R}_{m-1}})$ to change. If $V_i\neq{{\cal G}_{{\cal R}_{k}}}$ and  $V_i$ is before ${{\cal G}_{{\cal R}_{k}}}$, it follows that $\con({{\cal G}_{{\cal R}_{k'}}})$ could only expand and it is not possible for both c1 and c2 to be false, since $r_i \notin \CH({\cal G}_{{\cal R}_{k'}})$.Hence $r_i$ could have only cause $\con({{\cal G}_{{\cal R}_{m-1}}})$ to expand or did not caused any change because $\con({{\cal G}_{{\cal R}_{m-1}}})$ is bigger compared to $V_i$.
%Three possible events can happen:
%\begin{itemize}
%\item Stop($r_i$) or Collide($r_i$)
%$r_i$ moves a distance of at least $\delta$ with direction from $c_i$ to $p$, hence  it does not cause $\con({\cal G}_{{\cal R}_{m-1}})$ to change.
%\item Arrive($r_i$)
%$r_i$ moves to $p$, hence  it does not cause $\con({\cal G}_{{\cal R}_{m-1}})$ to change.
%\end{itemize}
\item $r_i$ cannot move towards $\CH({{\cal G}_{{\cal R}_{k'}}})$ without causing it to change.
Based on Lemma~\ref{NotTouchingLemma} and Function {\tt NotTouching}, $r_i$ gets into state {\bf Compute.ToChange}
and runs Procedure {\tt ToChange}. Based on Lemma ~\ref{ToChangeLemma}, the procedure returns a point $p \in \CH({\cal G}_{{\cal R}_{k'}})$ or $c_i$. If it is $c_i$, it follows that it does not cause $\con({\cal G}_{{\cal R}_{m-1}})$ to change.
Else, in the case $r_i$ does not arrive to $p$ (events {\em Stop or Collide}) then it follows that
it does not cause $\con({\cal G}_{{\cal R}_{m-1}})$ to change. In the case it arrives to $p$ (event {\em Arrive}) and $V_i\subseteq {{\cal G}_{{\cal R}_{k}}}$,
it is not difficult to see that $\con({\cal G}_{{\cal R}_{m-1}})$ can only increase (if for example
$\con({\cal G}_{{\cal R}_{k'}}) = \con({\cal G}_{{\cal R}_{m-1}})$, then $r_i$ it causes it to change,
but not to decrease). If $V_i\neq{{\cal G}_{{\cal R}_{k}}}$ and  $V_i$ is before ${{\cal G}_{{\cal R}_{k}}}$, it follows that $\con({{\cal G}_{{\cal R}_{k'}}})$ could only expand and it is not possible for both c1 and c2 to be false, since $r_i \notin \CH({\cal G}_{{\cal R}_{k'}})$.Hence $r_i$ could have only cause $\con({{\cal G}_{{\cal R}_{m-1}}})$ to expand or did not caused any change because $\con({{\cal G}_{{\cal R}_{m-1}}})$ is bigger compared to $V_i$.
%Three possible events can happen:
%\begin{itemize}
%\item Stop($r_i$) or Collide($r_i$)
%$r_i$ moves a distance of at least $\delta$ with direction from $c_i$ to $p$, hence  it does not cause $\con({\cal G}_{{\cal R}_{m-1}})$ to change.
%\item Arrive($r_i$)
%$r_i$ moves to $p$, hence  it does not cause $\con({\cal G}_{{\cal R}_{m-1}})$ to change.
%\end{itemize}
\end{itemize}
\end{itemize}
\end{itemize}

%b) c2 is true:
%This is a sub-case of a.

\item [(B)] $e_m$ involves indirectly a robot $r_j$ that is in state {\bf Move}.
This follows the same exact reasoning as with the case where $e_m$ involves
directly robot $r_i$ while in state {\bf Move}.\hfill$\blacksquare$
\end{enumerate}
\end{proof}

\begin{lemma}
\label{FullyVisibleLemma}
Given a good execution fragment ${\cal R}_x,e_x,\dots, {\cal R}_{m-1}$ such that $\forall {\cal R}_k$, $x \le k \le m-1$ holds that

 c1: $|\CH({\cal G}_{{\cal R}_{k}})|=n$ and ${\cal G}_{{\cal R}_{k}}$ is a $fully$ $visible$ configuration

AND

c2: ${\cal G}_{{\cal R}_{k}}$ is not a $connected$ configuration,

then for any step  $\langle {\cal R}_{m-1},e_m,  {\cal R}_{m}\rangle$, c1 holds for ${\cal G}_{{\cal R}_{m}}$ and $\con({\cal G}_{{\cal R}_{m-1}}) \supseteq \con({\cal G}_{{\cal R}_{m}})$

\end{lemma}

\begin{proof}
The possible events $e_{m}$ are:

\begin{enumerate}
\item $e_m$ involves (directly) robot $r_i$. If $r_i$ in ${\cal R}_{m-1}$ is in state {\bf Wait}, {\bf Look} or {\bf Compute}, then it trivially holds that none of the possible events $e_m$ can affect the $\con$. 
So, it remains to consider the case that $r_i$ in ${\cal R}_{m-1}$ is in state {\bf Move}.
In this case, there are three possible cases for event $e_m$: $stop(r_i),~arrive(r_i)$ or $collide(X), r_i\in X$.
\begin{enumerate}
\item [(A)]
Since $r_i$ is in state {\bf Move}, then it is following a trajectory ($start$,$target$), where
$start$ is the position of its center when it start moving, and $target$ is the position it wants to reach, as it was
calculated when the robot was in state {\bf Compute} (it is possible that $start = target$), say in ${\cal R}_k,~k<m-1$. 
Furthermore, $r_i$ made decisions based on the view the robot obtained while in state 
{\bf Look}, in some configuration $R_{k'},~k'<k$. It follows that $k'<m-1$, hence the lemma Hypothesis applies (i.e., for ${\cal G}_{{\cal R}_{k'}}$ properties c1 and c2 hold). In other words, for configurations $R_{k'}$ through $R_{m-1}$, $V_i \subseteq {\cal G}_{{\cal R}_{k'}}$. 
Now, for $V_i$ we have the following possible cases:

\begin{itemize}
\item Robot $r_i \in \CH({\cal G}_{{\cal R}_{k'}})$, $|\CH({\cal G}_{{\cal R}_{k'}})|=n$ and ${\cal G}_{{\cal R}_{k'}}$ is a $fully$ $visible$ configuration because c1 is true.

Based on Lemma ~\ref{StartLemma} and Function {\tt Start}, $r_i$ moves to state {\bf OnConvexHull}. Based on Lemma ~\ref{OnConvexHullLemma} and Function {\tt OnConvexHull}, $r_i$ moves to state{\bf AllOnConvexHull}. Based on Lemma ~\ref{AllOnConvexHullLemma} and Function {\tt AllOnConvexHull}, $r_i$ moves to state {\bf NotConnected}. Procedure {\tt NotConnected} returns a point $p \in \con({\cal G}_{{\cal R}_{k'}})$ . Three possible events can happen:
\begin{itemize}
\item Stop($r_i$) or Collide($r_i$)

$r_i$ moves a distance of at least $\delta$ with direction from $c_i$ to $p$. Because of Lemma ~\ref{NotConnectedLemma}, $r_i$ does not cause $|\CH({\cal G}_{{\cal R}_{k'}})|<n$ or ${\cal G}_{{\cal R}_{k'}}$ to be not a $fully$ $visible$ configuration. Because $p \in \con({\cal G}_{{\cal R}_{k'}})$ and $p \notin \CH({\cal G}_{{\cal R}_{k'}})$, it follows that $\con({\cal G}_{{\cal R}_{k}})$ can only shrink.

\item Arrive($r_i$)

$r_i$ moves to $p$. Because of Lemma ~\ref{NotConnectedLemma}, $r_i$ does not cause $|\CH({\cal G}_{{\cal R}_{k'}})|<n$ or ${\cal G}_{{\cal R}_{k'}}$ to be not a $fully$ $visible$ configuration. Because $p \in \con({\cal G}_{{\cal R}_{k'}})$ and $p \notin \CH({\cal G}_{{\cal R}_{k'}})$, it follows that $\con({\cal G}_{{\cal R}_{k}})$ can only shrink.
\end{itemize}

\item Otherwise

This case is not possible, since c1 is true.
\end{itemize}

Another robot $r_j$ could also was in state {\bf Move} in $e_{m}$. We get the following cases:

a) $r_j$ is in a trajectory ($start$,$target$), that was decided on a robot configuration, say ${\cal R}_k$. It follows that $x<k<m-1$, hence Lemma Hypothesis applies.
Specifically, $r_j$ had a view where c1 and c2 were true for ${\cal R}_k$.
This is the same case with $r_i$ (previous).

b) $r_j$ is in a trajectory ($start$,$target$), that was decided on a robot configuration, say ${\cal R}_k$. It follows that $k<x$.
\begin{itemize}
\item c1 and c2 in ${\cal R}_k$ were true

This is the same case with 1-A

\item c1 was not true in  ${\cal R}_k$. We get the following cases:
\begin{itemize}

\item $r_j \in \CH({\cal G}_{{\cal R}_{k}})$ 

Based on Lemma ~\ref{StartLemma} and Function {\tt Start}, $r_j$ moves to state {\bf Compute.OnConvexHull}. Based on Lemma ~\ref{OnConvexHullLemma} and Function {\tt OnConvexHull}, $r_j$ changes to state {\bf Compute.NotAllOnConvexHull}, because $|\CH({{\cal G}_{{\cal R}_{k}}})|<n$.

\begin{itemize}
\item $r_j$ is on straight line with any two other robots that $\in \CH({{\cal G}_{{\cal R}_{k}}})$

Based on Lemma ~\ref{NotAllOnConvexHullLemma} and Function {\tt NotAllOnConvexHull}, robot $r_j$ moves to state {\bf Compute.OnStraightLine}.

\begin{itemize}
\item $r_j$ is in the middle of two other robots that $\in \CH({{\cal G}_{{\cal R}_{k}}})$

Based on Lemma ~\ref{OnStraightLineLemma} and Function {\tt OnStraightLine}, robot $r_i$ moves to state {\bf Compute.SeeTwoRobots}. Based on 
Procedure {\tt SeeTwoRobots} and per Lemma ~\ref{SeeTwoRobotLemma}, it returns a point $p$ with direction away from the convex hull.

This case is not possible to happen since it is considered as bad configuration of Type 2.

\item  $r_j$ is not in the middle of two other robots that $\in \CH({{\cal G}_{{\cal R}_{k}}})$

Based on Lemma ~\ref{OnStraightLineLemma} and Function {\tt OnStraightLine}, robot $r_j$ moves to state {\bf Compute.SeeOneRobot}.
Procedure {\tt SeeOneRobot} based on Lemma ~\ref{SeeOneRobotLemma}, returns $c_j$, hence $r_j$ does not cause $\con({\cal G}_{{\cal R}_{m-1}})$ to change.

\end{itemize}
\item $r_j$ is not on straight line with any two other robots that $\in \CH({{\cal G}_{{\cal R}_{k}}})$

Based on Lemma ~\ref{NotAllOnConvexHullLemma} and Function {\tt NotAllOnConvexHull}, robot $r_j$ moves to state {\bf Compute.NotOnStraightLine}.

\begin{itemize}
\item $r_j$ sees that there exist enough space for at least one robot on $\CH({{\cal G}_{{\cal R}_{k}}})$

Based on Lemma ~\ref{NotOnStraightLineLemma} and Function {\tt NotOnStraightLine}, robot $r_j$ moves to state {\bf Compute.SpaceForMore}. Procedure {\tt SpaceForMore} based on Lemma ~\ref{SpaceForMoreLemma}, returns $c_j$ or $p$ a point outside of the convex hull if $r_j$ touches another not adjacent robot on $\CH({\cal G}_{{\cal R}_{m-1}})$ . If it returns $c_j$ $r_j$ does not cause $\con({\cal G}_{{\cal R}_{m-1}})$ to change. The case that $r_j$ touches another not adjacent robot on $\CH({\cal G}_{{\cal R}_{m-1}})$ 
is impossible because this means that the two robots that are touching block at least one robot from seeing other robots, hence it is impossible to have $fully$ $visible$ and this situation.

\item $r_j$ sees that not enough space exists for at least one robot on $\CH({{\cal G}_{{\cal R}_{k}}})$

Based on Lemma ~\ref{NotOnStraightLineLemma} and Function {\tt NotOnStraightLine}, robot $r_j$ moves to state {\bf Compute.NoSpaceForMore}. Procedure {\tt NoSpaceForMore} based on Lemma ~\ref{NoSpaceForMoreLemma}, returns a point $p$ with direction away from the convex hull.

This case is not possible to happen since it is considered as bad configuration of Type 1.

\end{itemize}

\end{itemize}

\item $r_j \notin \CH({{\cal G}_{{\cal R}_{k}}})$

Based on Lemma ~\ref{StartLemma} and Function {\tt Start}, $r_j$ moves to state {\bf Compute.NotOnConvexHull}

\begin{itemize}
\item $r_j$ is touching another robot.

Based on Lemma ~\ref{NotOnConvexHullLemma} and Function {\tt NotOnConvexHull}, $r_j$ moves to state {\bf Compute.IsTouching}. Procedure {\tt IsTouching} based on Lemma ~\ref{IsTouchingLemma}, returns a point $p \in \CH({\cal G}_{{\cal R}_{k}})$ or $c_j$.

This case is not possible to happen, because if $r_j$ did not arrived to p before $e_{m}$, it is not possible for $|\CH({\cal G}_{{\cal R}_{m-1}})|=n$, since no robot that belongs to $\CH({\cal G}_{{\cal R}_{m-1}})$ moves and neither does $r_j$.

\item $r_j$ is not touching any other robot.

Based on Lemma ~\ref{NotOnConvexHullLemma} and Function {\tt NotOnConvexHull}, $r_j$ moves to state {\bf Compute.NotTouching}. 

\begin{itemize}
\item $r_j$ can move to $\CH({{\cal G}_{{\cal R}_{k}}})$ without causing it to change

Based on Lemma ~\ref{NotTouchingLemma} and Function {\tt NotTouching}, $r_j$ moves to state {\bf Compute.NotChange}. Based on Lemma ~\ref{NotChangeLemma}, Procedure {\tt NotChange} returns a point $p \in \CH({\cal G}_{{\cal R}_{k}})$.

This case is not possible to happen, because if $r_j$ did not arrived to p before $e_{m}$, it is not possible for $|\CH({\cal G}_{{\cal R}_{m-1}})|=n$, since no robot that belongs to $\CH({\cal G}_{{\cal R}_{m-1}})$ moves and neither does $r_j$.

\item $r_j$ cannot move to $\CH({{\cal G}_{{\cal R}_{k}}})$ without causing it to change

Based on Lemma ~\ref{NotTouchingLemma} and Function {\tt NotTouching}, $r_j$ moves to state {\bf Compute.ToChange}. Based on Lemma ~\ref{ToChangeLemma}, Procedure {\tt ToChange} returns a point $p \in \CH({\cal G}_{{\cal R}_{k}})$ or $c_j$.

This case is not possible to happen, because if $r_j$ did not arrived to p before $e_{m}$, it is not possible for $|\CH({\cal G}_{{\cal R}_{m-1}})|=n$, since no robot that belongs to $\CH({\cal G}_{{\cal R}_{m-1}})$ moves and neither does $r_j$.
\end{itemize}

\end{itemize}

\end{itemize}

\end{itemize}

\item [(B)] $r_i$ is in a trajectory ($start$,$target$), that was decided on a robot configuration, say ${\cal R}_k$. It follows that $k<x$.

This is a similar case with $r_j$ in 1-A-b.
\end{enumerate}

\item $e_{m-1}$ on $r_j$ (indirect)

This is the same case with $r_j$ in 1-A-a and 1-A-b.\hfill$\blacksquare$
\end{enumerate}

\end{proof}

\subsubsection*{Liveness Properties}

\begin{lemma}
\label{ReachFullyVisibleLemma}
Given any good execution of the algorithm, there exists a configuration ${\cal R}_m$ such that $|\CH({\cal G}_{{\cal R}_m})|=n$ and ${\cal G}_{{\cal R}_m}$ is a fully visible configuration.
\end{lemma}

\begin{proof}
If ${\cal R}_0$ has the stated properties, there is nothing to prove. So consider the case that ${\cal R}_0$ is either c1: $|\CH({\cal G}_{{\cal R}_0})|<n$ or c2: $|\CH({\cal G}_{{\cal R}_0})|=n$ and ${\cal G}_{{\cal R}_0}$ is not a fully visible configuration.

Based on Lemma ~\ref{IncreaseLemma}, if c1 or c2 is true, then $\CH({\cal G}_{{\cal R}_0})$ can only expand, hence $\CH({\cal G}_{{\cal R}_0})$ will not shrink unless c1 and c2 are not true. 

We first list the various cases to be considered and then we show how they are interleaved.
\begin{enumerate}
\item c1 is true.
\begin{itemize}
\item[A] There exists space for at least one robot to be on the convex hull.
\begin{itemize}
\item[i] Robots that $\in\CH$

In this case, based on Lemmas ~\ref{SpaceForMoreLemma}, ~\ref{SeeOneRobotLemma} and ~\ref{SeeTwoRobotLemma},the robots that $\in \CH$ do not move or move outside of the convex hull.

\item[ii] Robots that $\notin \CH$
\begin{itemize}
\item[a] Robots that are tangent with other robots.

In this case, based on Lemma ~\ref{IsTouchingLemma}, Robots that are tangent with other robot either stay in the same position, or move to $\CH$.

\item[b] No point on $\CH$ exists, such that Function {\tt FindPoints} will return it as valid point.

In this case, based on Lemma ~\ref{ToChangeLemma}, robots that called Function {\tt FindPoints} and no point was returned, will move to $\CH$.

\item[c]  At least a point on $\CH$ exists, such that Function {\tt FindPoints} will return it as valid point.

In this case, based on Lemma ~\ref{NotChangeLemma}, robots that called Function {\tt FindPoints} and at least a point was returned, will move to $\CH$.
\end{itemize}
\end{itemize}
\item[B] No space exists for at least one robot on $\CH$.

\begin{itemize}

\item[i] Robots that $\notin \CH$.

In this case, based on Lemmas ~\ref{IsTouchingLemma}, ~\ref{ToChangeLemma} and ~\ref{NotChangeLemma}, robots that $\notin \CH$ do not move.

\item[ii] Robots that $\in \CH$.

In this case, robots that $\in \CH$, Based on Lemmas ~\ref{NoSpaceForMoreLemma}, ~\ref{SeeOneRobotLemma} and ~\ref{SeeTwoRobotLemma} can only move with direction outside of the convex hull or stay at the same position. 

\end{itemize}

\end{itemize}

\item c2 is true.

In this case $|\CH({\cal G}_{{\cal R}_0})|=n$ and ${\cal G}_{{\cal R}_0}$ is not a fully visible configuration. This implies that at least three robots are on the same line, hence we get the following cases:

\begin{itemize}

\item[A] Robots that are not on a straight line with any two other robots.

In this case, based on Lemma ~\ref{SpaceForMoreLemma}, robots stay in the same position.

\item[B] Robots that are on the same straight line with at least two other robots but are not in the middle of any two other robots that are on the same line.

In this case, based on Lemma ~\ref{SeeOneRobotLemma}, robots stay in the same position.

\item[C] Robots that are on the same straight line with at least two other robots and are in the middle of any two other robots that are on the same line.

In this case, based on Lemma ~\ref{SeeTwoRobotLemma}, robots move outside of the convex hull.
\end{itemize}
\end{enumerate}

We now discuss how the cases above are combined to yield the claimed result.\vspace{.4em}

(a) If no space for at least one robot on the convex hull exist, this is case 1-B. In case 1-B, necessary some robots are on the convex hull and this is case 1-B-ii for some robots. Therefore, robots of case 1-B-ii will continue to expand until a space for at least one robot exists. Hence if a space does not exist, eventually a space for more robots on the convex hull will be created.

(b)If some robots that are touching are in case 1-A-ii-a, based on Lemma ~\ref{IsTouchingLemma2}, at least one robot will move. Hence, eventually the robots that were tangent will no longer be tangent in the same place.

(c) If three or more robots are on the same line, it means that at least one robot is in the middle of two other robots. The robot that is in the middle, based on Lemma ~\ref{SeeTwoRobotLemma}, will move to the outside of the convex hull. Each time the robots that are not in the middle, Based on Lemma ~\ref{SeeOneRobotLemma} will stay in the same position. Therefore, eventually no three robots will be on the same line and each time there exists a line, the convex hull expands.

(d)If c1 is true it means that at least one robot is not on the convex hull. If a space for at least one robot on the convex hull exists, then it could be one of the cases 1-A-ii. Robots in cases 1-A-ii-a (at least 1), 1-A-ii-b and 1-A-ii-c try to move on the convex hull. If at least one space on the convex hull exists, one of the robots that are inside the convex hull will move to $\CH$. Because of (c) eventually no three robots will be on the same line, hence the robots on the convex hull will be run Procedure {\tt NoSpaceForMore} (see the possible cases if c1 is true and no 3 robots are on the same line). If no space exists on the convex hull robots that are on the convex hull will move to expand to the convex hull and create more space as described earlier in (a) .
Hence if c1 is true it follows that the convex hull expands. 

(e) Based on Lemma ~\ref{CHlengthLemma}, for any two adjacent robots with centers $c_l$ and $c_r$, $c_l$ and $c_r\in \CH$, there
exists a $safe$ $distance$ between $c_l$ and $c_r$ for which a third robot $r_i$ can be on $\CH$ between $c_l$ and $c_r$ without causing it to change.

(f) Based on (d) and (e) it follows that if c1 is true convex hull will continue expanding and the number of robots that are on the convex hull will increase, until c1 is not true or the $safe$ $distance$ was reached. Some robots that get on the convex hull cause some other robots to no longer be on the convex hull. This means that the convex hull will continue to expand if c1 is true and after a safe distance between 
neighboring robots on the convex hull is reached, the next robot that is inside the convex hull can and will move on the convex hull without causing it to change. This will continue happening until all robots are on the convex hull. Hence c2 will be true.

(g) If c2 is true, it means that at least three robots are on the same line. Based on (c) the convex hull expands and eventually no three robots  will be on the same line. Some robots that move to the outside of the convex hull may cause others to no longer be on the convex hull. Then c1 will be true and based on (f) c2 will be true again. This will continue happening until safe distance is reached (The convex hull continues expanding if c1 or c2 is true). In the same way as in (f) all robots will be on the convex hull without any changes caused and based on (c) no three robots will be on the same line. Hence, all robots will be on the convex hull and all robots will have full visibility. This completes the proof.
\hfill$\blacksquare$\vspace{1em}
\end{proof}

The following lemma states that starting from any initial configuration, 
when the robots form a configuration such that all robots are on the
convex hull and they can see each other, then the robots will eventually
form a connected configuration.

\begin{lemma}
\label{ReachConnectedLemma}
Given any good execution of the algorithm, if ${\cal R}_l$ is such that $|\CH({\cal G}_{{\cal R}_l})|=n$ and  ${\cal G}_{{\cal R}_l}$ is a fully visible configuration and not a $connected$ configuration, then there exists ${\cal R}_k$, $l\le k$ so that ${\cal R}_k$ is a connected configuration.
\end{lemma}

\begin{proof}
Based on Lemma ~\ref{FullyVisibleLemma}, if a configuration ${\cal R}_m$ is such that $|\CH({\cal G}_{{\cal R}_m})|=n$ and  ${\cal G}_{{\cal R}_m}$ is a $fully$ $visible$ configuration, then $|\CH({\cal G}_{{\cal R}_{m+1}})|=n$,  ${\cal G}_{{\cal R}_{m+1}}$ is a $fully$ $visible$ configuration and $\con({\cal G}_{{\cal R}_m}) \subseteq \con({\cal G}_{{\cal R}_{m+1}})$. 

Based on Procedure {\tt NotConnected} (see first three cases of procedure), no robot will start moving unless:
Between any three adjacent robots on the convex hull, say $r_l,r_m$ and $r_r$ left robot, middle robot and right robot respectively, the distance between line segment $\overline{r_lr_r}$ and $r_m$ must be equal or more than $\frac{1}{n}$.
This, along with Lemma ~\ref{ReachFullyVisibleLemma}  guarantee that no robot will move unless the distance of $\frac{1}{n}$ at least exists and that eventually all robots will be on the convex hull and have full visibility. Because no robot moves unless the distance of $\frac{1}{n}$ at least exists, all robots will eventually move to the {\bf Look} state and see that the configuration they see is $fully$ $visible$ and $|\CH(V_i)|=n$.
We get the three following cases:
\begin{itemize}
\item[A] There exists at least one component (as it was defined in Function ~\ref{subsec:RC}) that is smaller than at least one other component, with respect to the number of the robots that consist each component

Function {\tt NotConnected} results all robots of the smallest component(s) to join one component that is larger than it. Given the liveness condition that whenever a robot decides to move, it moves at least a distance of $\delta$, eventually the number of the components become smaller and eventually the convex hull shrinks. Also the robots, of the components that are not the smallest, do not move.

\item[B] All components are of the same size, with respect to the number of the robots that consist each component. The distance between two neighboring components is not the same for all the neighboring components.

Function {\tt NotConnected} results that all robots of the component that has the smallest distance to its neighbor component on the right to join the component on its right. Given the liveness condition that whenever a robot decides to move, it moves at least a distance of $\delta$, eventually the number of the components become smaller and eventually the convex hull shrinks. The robots of the other components do not move.

\item[C] All components are of the same size, and the distance between any two neighboring components is the same.

Function {\tt NotConnected} results that all the components start moving with direction to the inside of the convex hull. Given the liveness condition that whenever a robot decides to move, it moves at least a distance of $\delta$, it follows that eventually all the components will touch, because the convex hull shrinks. 

\end{itemize}

From the cases above, it follows that either all the robots of any component that has the smallest number of robots (first case) or of any component that has the smallest distance (second case) to its right neighbor will move to its right neighbor until the number of components become one, or the components will move to the inside of the convex hull until all the components touch (third case).

In every case, robot $r_i$ runs the Procedure {\tt NotConnected}. Hence, per 
Lemma~\ref{NotConnectedLemma}, robot $r_i$ moves in such a way that it does not cause $|\CH({\cal G}_{{\cal R}_{m+1}})|<n$ or  ${\cal G}_{{\cal R}_{m+1}}$ not to be a $fully$ $visible$ configuration. This completes the proof.\hfill$\blacksquare$\vspace{1em}
\end{proof}

From Lemmas~\ref{ReachFullyVisibleLemma} and \ref{ReachConnectedLemma} we get the following.

\begin{corollary}
\label{ReachConnectedFullCorollary}
Given any good execution of the algorithm, there exists ${\cal R}_m$ so that ${\cal G}_{{\cal R}_m}$ is a connected and fully visible configuration.
\end{corollary}

%\begin{lemma}
%\label{StayLemma}
%If robot $r_i$ reach state $\bot$, the state of robot $r_i$ will be $\bot$ in every future execution and all robots will reach state $\bot$ in their next cycle.
%\end{lemma}

\subsubsection{Any Execution}

%%%CG if it is needed
%We define as {\em dangerous area} any area on $\CH$ between any two neighboring robots that $\in \CH$, where the adversary, 
%with its strategy can cause a bad configuration.
We now consider any executions, including bad ones. 

\begin{lemma}
\label{ReachGoodConfigurationLemma}
Given any execution of the algorithm, if there is a bad execution fragment $\alpha_{bad}$, then eventually a $safe$ configuration ${\cal R}_{safe}$ is reached, and after  a $safe$ configuration there are no longer any bad configurations in the execution until termination. 
\end{lemma}

\begin{proof}
There are 2 possible cases:\\
(a) The adversary deploys a strategy that aims in causing bad configurations as long as it can
(i.e., indefinitely if possible).\\
(b) The adversary, at some point of the execution, stops causing bad configurations. 

We focus on the first case and we show that any execution under this adversarial strategy will eventually
reach a configuration in which the adversary will no longer be able to cause bad configurations. It is easy to
see that this case covers also the second case.

Recall that both types of bad configurations involve configurations in which the robots
are momentarily in a configuration in which all robots are on the convex hull and it is
fully visible, but the adversary manages to break this property. The adversary, as explained,
exploits the fact that some robots, due to asynchrony, are not aware that such a configuration
has been reached. We now consider the two types of bad configurations.

\noindent {\bf (i) Bad configuration of type 1.}
Consider the case in which the first bad configuration, call it ${\cal R}_x$, that appears in the bad execution fragment $\alpha_{bad}$ 
is of type 1 (the other type is considered later). As explained, the adversary may deploy a strategy which can result into a configuration
${\cal R}_z,~z > x$, so that ${\cal G_R}_z$ is no longer fully visible or/and not all robots are on the convex hull. The adversary
can do so, if there is at least one robot that according to its local view in configuration ${\cal R}_x$, not all robot are on the 
convex hull and there is no more space for an ``internal'' robot to get on the convex hull (per Function {\tt NoSpaceForMore} 
this robot will move to a direction outside of the convex hull).  It follows that $\con({\cal G_R}_z) \supseteq \con({\cal G_R}_x)$. Furthermore, from Lemma~\ref{IncreaseLemma} we get that
for all successive configurations of ${\cal R}_z$ in which not all robots are on the convex hull or are fully visible, the convex hull can 
only expand (until a configuration in which these properties hold is reached). The adversary may repeat this strategy (e.g., involving other robots
on the convex hull), every time causing the convex 
hull to expand. However, per Lemma~\ref{CHlengthLemma}, this cannot be repeated indefinitely, as the convex hull will expand that much, that the {\em safe} distance will be reached for all pairs of adjacent robots on the convex hull. From this and the liveness condition (the adversary must allow a robot to move by at least $\delta$ distance) it follows that a configuration is eventually reached after which no bad configuration of type 1 can exist (no robot will get into state {\bf Compute.NoSpaceForMore}). Observe that when such a configuration is reached, it is still possible for 
a bad configuration of type 2 to be reached. This is covered by the next case we consider (with the difference that this bad configuration is not the first appearing in $\alpha_{bad}$). 
 
\noindent {\bf (ii) Bad configuration of type 2.} Consider the case in which the first bad configuration, call it ${R}_x$, that appears in the bad execution fragment $\alpha_{bad}$ is of type 2. This is the situation where in a preceding configuration there are at least four robots on a
straight line on the convex hull. As explained in Section~\ref{baddefs}, the adversary can yield a configuration in which not all robots
are any longer on the convex hull, or there is no full visibility. However, per Function {\tt SeeOneRobot} and Lemma~\ref{SeeOneRobotLemma} 
the robots on the straight line that are not in the middle (i.e., they see only one robot) do not move. In contrast, according to Function {\tt SeeTwoRobot} and Lemma~\ref{SeeTwoRobotLemma}, each robot in the middle of the straight line moves in a direction outside of the convex hull, 
in such a way that it will no longer be in a straight line with its two adjacent robots (on the convex hull). It follows that if every time the adversary repeats the same strategy, and say initially there are $x$ robots on straight line, then in every iteration the number of robots that 
are on the same line is $x-2$. This may continue only until x is less than 3, hence it eventually stops. Observe that during these iterations, 
since robots in the middle move towards a direction outside of the convex hull and per Lemma~\ref{IncreaseLemma}, the convex hull can only
expand. Hence a bad configuration of type 2 can no longer exist. Furthermore, note that if during this expansion, the robots involved 
have also reached the safe distance (per Lemma~\ref{CHlengthLemma}'s definition), then as explained above, a bad configuration of type 1
also cannot exist. Otherwise, we are back in case (i) as discussed above. Note however that once robots reach the safe distance, and a bad configuration of type 2 is reached, a configuration of type 1 can no longer exist again: when a robot has already safe distance between its adjacent robots on the convex hull, then the middle robots by moving towards outside the convex hull can only increase the safe distance (and hence
it will not be possible for a robot to get into state {\bf Compute.NoSpaceForMore}). 

From cases (i) and (ii) and Lemma~\ref{ReachFullyVisibleLemma} it follows that a fully visible configuration in which $|\CH|=n$ is reached.
By a similar argument as in the proof of Lemma~\ref{ReachConnectedLemma} we get that eventually a safe configuration is reached (all robots are on the convex and they are aware that the configuration is fully visible). From Function {\tt NotConnected} and Lemma~\ref{NotConnectedLemma} it follows that any succeeding configuration maintains the property that all robots can see each other and that are on the convex hull. Hence, the algorithm is such that once a safe configuration is reached, it is no longer possible for a bad configuration to exist. This completes the proof.\hfill$\blacksquare$\vspace{1em}
\end{proof}

We are now ready to prove that our algorithm solves the gathering problem.

\begin {theorem}[Gathering]
In any execution of algorithm, there exists a configuration ${\cal R}_m$, so that ${\cal G}_{{\cal R}_m}$ is a connected, fully visible configuration and $\forall s_i \in {\cal S}_{{\cal R}_m}$, $s_i=\mathbf{Terminate}$.
\end{theorem}

\begin{proof}
Consider the following two cases.
\begin{itemize}
\item If no bad configurations exist, based on Corollary ~\ref{ReachConnectedFullCorollary}, given any good execution of the algorithm, there exists ${\cal R}_m$ so that ${\cal G}_{{\cal R}_m}$ is a connected and fully visible configuration. 
\item If bad configurations exist, based on Lemma ~\ref{ReachGoodConfigurationLemma}, given any execution of the algorithm, if there is a bad execution fragment $\alpha_{bad}$, then eventually a $safe$ configuration ${\cal R}_{safe}$ is reached, and after  a $safe$ configuration there are no longer any bad configurations in the execution until termination. Therefore, from this point onward, we get from Corollary ~\ref{ReachConnectedFullCorollary} that there exists ${\cal R}_m$ so that ${\cal G}_{{\cal R}_m}$ is a connected and fully visible configuration.

\end{itemize}

When a $connected$ and $fully$ $visible$ configuration is reached, it is easy to see that robots no longer move and eventually all robots get into state {\bf Compute.Connected} and hence into state {\bf Terminate}.\hfill$\blacksquare$ 
\end{proof}

%% file: Fat-Robots-CA.bbl
\begin{thebibliography}{12}


\bibitem{Peleg04}
N. Agmon and D. Peleg.
\newblock Fault-tolerant gathering algorithms for autonomous mobile robots.
\newblock In {\em Proc. of the 15th ACM-SIAM Symposium on Discrete Algorithms (SODA 2004)}, pages 1070--1078.

\bibitem {Japan99}
H. Ando, Y. Oasa, I. Suzuki, and M. Yamashita.
\newblock Distributed memoryless point convergence algorithm for mobile robots with limited visibility.
\newblock {\em IEEE Transactions on Robotics and Automation}, 15(5):818--828, 1999.

\bibitem{Attiya04}
H. Attiya and J. Welch.
\newblock {\em Distributed Computing: Fundamentals, Simulations
and Advanced Topics.}
\newblock Second edition, Wiley \& Sons, 2004.

\bibitem{Survey1}
A. Bandettini, F. Luporini, and G. Viglietta.
\newblock A survey on open problems for mobile robots.
\newblock In arXiv:1111.2259v1, 2011. 

\bibitem{Bolla12}
K. Bolla, T. Kov�cs, and G. Fazekas.
\newblock Gathering of fat robots with limited visibility and without global navigation. 
\newblock In {\em Proc. of ICAISC/SIDE-EC 2012}, pages 30--38.

\bibitem{Chaudhuri10}
S.G. Chaudhuri and K. Mukhopadhyaya.
\newblock Gathering asynchronous transparent fat robots.
\newblock In {\em Proc. of the 6th International Conference on Distributed Computing and Internet Technology (ICDCIT 2010)}, pages 170--175.

\bibitem{Friedhelm2011}
A. Cord-Landwehr, B. Degener, M. Fischer, M. H\"{u}llmann, B. Kempkes, A. Klaas, P. Kling, S. Kurras, M. M\"{a}rtens, F.M.A Der Heide, C. Raupach, K. Swierkot, D. Warner, C. Weddemann, and D. Wonisch.
\newblock Collisionless gathering of robots with an extent.
\newblock In {\em Proc. of the 37th International Conference on Current Trends in Theory and Practice of Computer Science (SOFSEM 2011)},
 pages 178--189.

\bibitem{Pelc09}
J. Czyzowicz, L. Gasieniec, and A. Pelc.
\newblock Gathering few fat mobile robots in the plane.
\newblock {\em Theoretical Computer Science,} 410(6--7):481--499, 2009.

\bibitem{Dutta12}
A. Dutta, S. G. Chaudhuri, S. Datta, and K. Mukhopadhyaya.
\newblock Circle formation by asynchronous fat robots with limited visibility. 
\newblock In {\em Proc. of the 8th International Conference on Distributed Computing and Internet Technology (ICDCIT 2012)}, pages 83--93.

\bibitem{robotsbook}
P. Flocchini, G. Prencipe, N. Santoro.
\newblock {\em Distributed Computing by Oblivious Mobile Robots}.
Synthesis Lectures on Distributed Computing Theory, Morgan \& Claypool Publishers, 2012.

\bibitem{Flocchini05}
P. Flocchini, G. Prencipe, N. Santoro, and P. Widmayer.
\newblock Gathering of asynchronous robots with limited visibility.
\newblock {\em Theoretical Computer Science}, 337(1--3):147--168, 2005.

\bibitem{Graham72}
R.L. Graham.
\newblock An efficient algorithm for determining the convex hull of a 
finite planar set.
\newblock {\em Information Processing Letters,} 1(4):132--133, 1972.

\bibitem{Survey2}
S. Souissi, T. Izumi, and K. Wada.
\newblock Distributed algorithms for cooperative mobile robots: A survey.
\newblock In {\em Proc. of the 2nd Second International Conference on Networking and Computing (ICNC 2011)}, pages 364--371.


\end{thebibliography}
